%% file: Article.tex
\documentclass[a4paper,UKenglish,cleveref,autoref,thm-restate]{lipics-v2021}

\input{Format}

\input{Macros}

\input{Figures}

\input{Tables}

\input{Algorithms}

\nolinenumbers

\hypersetup
  {
  pdftitle  = {Priority Promotion with Parysian Flair},
  pdfauthor = {M. Benerecetti, D. Dell'Erba, F. Mogavero, S. Schewe, D. Wojtczak}
  }

\title{Priority Promotion with Parysian Flair}
\titlerunning{Priority Promotion with Parysian Flair}

\author{Massimo Benerecetti}{Universit\`a degli Studi di Napoli Federico II, Naples, Italy}{massimo.benerecetti@unina.it}{https://orcid.org/0000-0003-4664-6061}{}
\author{Daniele Dell'Erba}{University of Liverpool, Liverpool, UK}{daniele.dell-erba@liverpool.ac.uk}{https://orcid.org/0000-0003-1196-6110}{}
\author{Fabio Mogavero}{Universit\`a degli Studi di Napoli Federico II, Naples, Italy}{fabio.mogavero@unina.it}{https://orcid.org/0000-0002-5140-5783}{}
\author{Sven Schewe}{University of Liverpool, Liverpool, UK}{sven.schewe@liverpool.ac.uk}{https://orcid.org/0000-0002-9093-9518}{}
\author{Dominik Wojtczak}{University of Liverpool, Liverpool, UK}{d.wojtczak@liverpool.ac.uk}{https://orcid.org/0000-0001-5560-0546}{}

\authorrunning{Benerecetti \etal}


\category{}

\hideLIPIcs

\begin{document}

%   \title{Priority Promotion with Parysian Flair}
%   \titlerunning{Priority Promotion with Parysian Flair}
% 
%   \author{}
%    \author{%
%      Massimo Benerecetti\inst{1}\orcidID{0000-0003-4664-6061} \and
%      Daniele Dell'Erba\inst{2}\orcidID{0000-0003-1196-6110} \and
%      Fabio Mogavero\inst{1}\orcidID{0000-0002-5140-5783} \and \\
%      Sven Schewe\inst{2}\orcidID{0000-0002-9093-9518} \and
%      Dominik Wojtczak\inst{2}\orcidID{0000-0001-5560-0546}}
% 
%    \authorrunning{Benerecetti \etal}
% 
%   \institute{}
%    \institute{%
%      Universit\`a degli Studi di Napoli Federico II, Naples, Italy \\
%      \email{\{massimo.benerecetti, fabio.mogavero\}@unina.it}
%      \and
%      University of Liverpool, Liverpool, UK \\
%      \email{\{daniele.dell-erba,sven.schewe,d.wojtczak\}@liverpool.ac.uk}}

  \maketitle

\input{Abstract}

  \sloppy

\input{Introduction}

\input{Preliminaries}

\input{PrioProm}

\input{Uncert}

\input{Discussion}

  % \input{Conclusion}

\input{Acknowledgments}

  \clearpage
  \bibliographystyle{plain}
  \bibliography{References,Bib}

  \clearpage
  \appendix

\input{AppendixB}

\input{AppendixA}

\input{AppendixC}

\input{ExperEval}

\end{document}

%% file: Format.tex
\usepackage[nofnttls,noenmtls,noamsthm,nothmtls,chgbar,alg]{fmocdmac}

\usepackage{paralist}

\AtEndPreamble
  {

  }

%% file: Macros.tex
% ...

\cmdtxtoparname{CPP}[C++]

\cmdtxtoparname{PP}
\cmdtxtoparname{NPP}
\cmdtxtoparname{RPP}
\cmdtxtoparname{HPP}

\cmdtxtoparname{ZLK}
\cmdtxtoparname{ZLKQ}

\cmdtxtoparname{QPT}

\cmdtxtoparname{SSPM}

\def\plrsym{0}
\def\oppsym{1}

\cmdmthsym{bot}[\bot]
\cmdmthsym{top}[\top]
\cmdmthsym{topZ}[\top_{\PlrSym}]
\cmdmthsym{topO}[\top_{\OppSym}]

\cmdmthset{Prm}[Pm]
\cmdmthfun{prm}[r]

\cmdmthset{QDom}[Qs]
\cmdmthfun{qdom}[r]

\cmdmthset{Reg}[Rg]
\cmdmthfun{reg}[r]

\def\sttset{St}
\def\sttsym{s}

\def\solfun{\mathbf{sol}}

\cmdmthfun{und}[u]

\cmdmthoargfun{bep}
\cmdmthoargfun{hsol}[\mathbf{hsol}]
\cmdmthoargfun{Nxt}[NextPr]
\cmdmthoargfun{Max}[Maximise]
\cmdmthoargfun{Prm}[Promote]
\cmdmthoargfun{Half}
\cmdmthoargfun{Und}

%% file: Figures.tex
\usepackage{tikz}
\usetikzlibrary{arrows,shapes,calc,patterns}

\usepackage{wrapfig}

\usepackage{pgf}
\usepackage{pgfplots}
\usepackage{pgfplotstable}
\pgfplotsset{compat=1.16}

\usepackage{multirow}

\usepackage{rotating}

\tikzstyle{every node} =
  [draw = none, fill = white, thin]
\tikzstyle{every edge} +=
  [black, thick]

\tikzstyle{noall} =
  [draw = none, fill = none]
\tikzstyle{nodraw} =
  [draw = none, fill = white]
\tikzstyle{nofill} =
  [draw = black, fill = none]

\tikzstyle{cnode} =
  [circle, draw = gray!50]
\tikzstyle{snode} =
  [regular polygon, regular polygon sides = 4, draw = gray!50]
\tikzstyle{lnode} =
  [diamond, draw = gray!75]
\tikzstyle{pnode} =
  [regular polygon, regular polygon sides = 5, draw = gray]

\AfterEndPreamble
  {

  \newcommand{\figsttspcexp}
    {
    \begin{tikzpicture}[>=stealth']

      \draw[very thick] (0, 0) rectangle (4, 4.625);

      \draw[thick, color = blue!50!black] (0.25, 3.75) rectangle (1.875, 4.375);
      \draw[thick, color = red!50!black] (2.125, 3.75) rectangle (3.75, 4.375);

      \draw[dashed] (0, 3.625) to (4, 3.625);

      \draw[thick, color = red!50!black] (0.25, 2.875) rectangle (3.75, 3.5);
      \draw[thick, color = blue!50!black] (0.25, 2.125) rectangle (3.75, 2.75);
      \draw[thick, color = red!50!black] (0.5, 1.125) rectangle (3.5, 1.75);
      \draw[very thick, color = gray] (0.25, 0.25) rectangle (3.75, 2);

      \draw (0.0625, 4.725) node [] {\large $\GameName$};

      \draw (0.325, 4.35) node [] {\scriptsize $\topZSym$};
      \draw (2.225, 4.35) node [] {\scriptsize $\topOSym$};

      \draw (0.25, 3.425) node [] {\scriptsize $7$};
      \draw (0.25, 2.625) node [] {\scriptsize $6$};
      \draw (0.5, 1.75) node [] {\scriptsize $3$};
      \draw (0.625, 1.175) node [] {\scriptsize $\RSet[\sttElm]$};
      \draw (0.35, 0.3) node [] {\small $\LSet[\sttElm]$};

      \draw (1.0625, 4.0625) node [] {\footnotesize $\gSym$};
      \draw (2.9475, 4.0625) node [] {\footnotesize $\emptyset$};

      \draw (2, 3.175) node [] {\footnotesize $\bSym$};
      \draw (2, 2.425) node [] {\footnotesize $\dSym, \fSym, \hSym$};
      \draw (2, 1.425) node [] {\footnotesize $\eSym$};
      \draw (2, 0.625) node [] {\footnotesize $\aSym, \cSym$};

    \end{tikzpicture}
    }

  \newcommand{\figsttspchyb}
    {
    \begin{tikzpicture}[>=stealth']

      \draw[very thick] (0, 0) rectangle (4, 4.75);

      \draw[thick, color = blue!50!black] (0.25, 3.875) rectangle (1.875, 4.5);
      \draw[thick, color = red!50!black] (2.125, 3.875) rectangle (3.75, 4.5);

      \draw[dashed] (0, 3.75) to (4, 3.75);

      \draw[thick, color = brown] (2, 3.625) rectangle (3.875, 1);
      \draw[thick, dashed, color = brown] (2, 1.95) to (3.875, 1.95);

      \draw[thick, color = red!50!black] (0.25, 2.875) rectangle (1.875, 3.5);
      \draw[thick, color = blue!50!black] (2.125, 2.875) rectangle (3.75, 3.5);

      \draw[thick, color = blue!50!black] (0.25, 2.125) rectangle (1.875, 2.75);
      \draw[thick, color = red!50!black] (2.125, 2.125) rectangle (3.75, 2.75);

      \draw[thick, color = red!50!black] (0.375, 1.125) rectangle (1.75, 1.75);
      \draw[thick, color = blue!50!black] (2.125, 1.125) rectangle (3.75, 1.75);

      \draw[very thick, color = gray] (0.25, 0.25) to (0.25, 1.95) to (1.875,
        1.95) to (1.875, 0.875) to (3.75, 0.875) to (3.75, 0.25) to (0.25,
        0.25);

      \draw (0.0625, 4.85) node [] {\large $\GameName$};

      \draw (0.325, 4.475) node [] {\scriptsize $\topZSym$};
      \draw (2.225, 4.475) node [] {\scriptsize $\topOSym$};

      \draw (0.35, 0.3) node [] {\small $\LSet[\sttElm]$};

      \draw (1.0625, 4.1875) node [] {\footnotesize
        $\regFun$};
      \draw (2.9475, 4.1875) node [] {\footnotesize
        $\regFun$};

      \draw (1.0625, 3.175) node [] {\footnotesize $\regFun$};
      \draw (1.0625, 2.425) node [] {\footnotesize $\regFun$};

      \draw (2.925, 3.175) node [] {\footnotesize $\undFun$};
      \draw (2.925, 2.425) node [] {\footnotesize $\undFun$};

      \draw (0.555, 1.225) node [] {\scriptsize $\RSet[\sttElm]$};
      \draw (3.575, 1.25) node [] {\scriptsize $\USet[\sttElm]$};

      \draw (1.0625, 1.425) node [] {\footnotesize $\regFun$};
      \draw (2.925, 1.425) node [] {\footnotesize $\undFun$};

      \draw (2, 0.55) node [] {\footnotesize $\regFun$};

    \end{tikzpicture}
    }

  \newcommand{\figexma}
    {
    \begin{tikzpicture}
      [node distance = 5em, bend angle = 22.5, inner sep = 0.20em, minimum size
      = 3em]

      \tikzset{every loop/.style = {max distance = 1.5em}}

      \node [cnode]
            (A)
            []
            {$\aSym/\mathbf{0}$};
      \node [snode]
            (B)
            [right of = A]
            {$\bSym/\mathbf{7}$};
      \node [snode]
            (C)
            [below of = A]
            {$\cSym/\mathbf{1}$};
      \node [cnode]
            (D)
            [right of = C]
            {$\dSym/\mathbf{5}$};
      \node [cnode]
            (E)
            [below of = C]
            {$\eSym/\!\mathbf{3}$};
      \node [snode]
            (F)
            [right of = E]
            {$\fSym/\mathbf{6}$};
      \node [cnode]
            (G)
            [below of = E]
            {$\gSym/\mathbf{4}$};
      \node [cnode]
            (H)
            [right of = G]
            {$\hSym/\mathbf{2}$};

      \path[->]
        (A) edge  []
                  (B)
            edge  [bend left]
                  (C)
        (B) edge  []
                  (D)
        (C) edge  [bend left]
                  (A)
            edge  [bend angle = 39.75, bend right]
                  (G)
        (D) edge  [loop right]
                  ()
            edge  [very thick, blue, dashed, bend left]
                  (F)
        (E) edge  []
                  (C)
        (F) edge  []
                  (E)
            edge  [bend left]
                  (D)
            edge  [bend left]
                  (H)
        (G) edge  [loop left,very thick, blue, dashed]
                  ()
        (H) edge  [very thick, blue, dashed, bend left]
                  (F)
            edge  []
                  (G)
        ;

       \begin{scope}
         [very thick, draw = red, fill = gray, fill opacity = 0.075, dashed]
         \filldraw
            ($(B) + (0, 0.75)$)
              to [out = 180, in = 90]
            ($(B) + (-0.75, 0)$)
              to [out = 270, in = 180]
            ($(B) + (0, -0.75)$)
              to [out = 0, in = 270]
            ($(B) + (0.75, 0)$)
              to [out = 90, in = 0]
            ($(B) + (0, 0.75)$)
            ;
        \filldraw
            ($(E) + (0, 0.75)$)
              to [out = 180, in = 90]
            ($(E) + (-0.75, 0)$)
              to [out = 270, in = 180]
            ($(E) + (0, -0.75)$)
              to [out = 0, in = 270]
            ($(E) + (0.75, 0)$)
              to [out = 90, in = 0]
            ($(E) + (0, 0.75)$)
            ;
       \end{scope}

       \begin{scope}
         [very thick, draw = blue, fill = gray, fill opacity = 0.075]
         \filldraw
            ($(G) + (0, 0.80)$)
              to [out = 180, in = 90]
            ($(G) + (-0.85, 0)$)
              to [out = 270, in = 180]
            ($(G) + (0, -0.80)$)
              to [out = 0, in = 270]
            ($(G) + (0.80, 0)$)
              to [out = 90, in = 0]
            ($(G) + (0, 0.80)$)
            ;
       \end{scope}

       \begin{scope}
         [very thick, draw = blue, fill = gray, fill opacity = 0.075, dashed]
         \filldraw
           ($(D) + (0, 0.80)$)
             to [out = 180, in = 90]
           ($(D) + (-0.75, 0)$)
             to [out = 0, in = 90]
           ($(D) + (-0.75, 0)$)
             to [out = 270, in = 90]
           ($(F) + (-0.75, 0)$)
             to [out = 270, in = 90]
           ($(H) + (-0.75, 0)$)
             to [out = 270, in = 180]
           ($(H) + (0, -0.80)$)
             to [out = 0, in = 270]
           ($(H) + (0.85, 0)$)
             to [out = 90, in = 270]
           ($(F) + (0.85, 0)$)
             to [out = 90, in = 270]
           ($(D) + (0.85, 0)$)
              to [out = 90, in = 0]
           ($(D) + (0, 0.80)$)
           ;
       \end{scope}

    \end{tikzpicture}
    }

  \newcommand{\figperthreeR}
		{
		\begin{figure}
			\begin{center}
 				\begin{tikzpicture}[scale = 0.7, every node/.style = {scale = 0.675}]
					\begin{axis}
						[
							name = lowerpart,
							width = 0.60\textwidth, height = 0.10\textwidth,
							xmin = 1, xmax = 10, ymin = 0, ymax = 5,
							scale only axis,
							axis x line* = bottom,
              xtick={1,2,4,6,8,10},
              xticklabels = {$10^4$,$2\cdot 10^4$,$4\cdot 10^4$,$6\cdot 10^4$,$8\cdot 10^4$,$10^5$},
							x axis line style = -, y axis line style = -,
							ytick = {0, 1, ..., 5},
							extra y ticks = 5,
							extra y tick style = {grid style = {black, dashed}},
							extra y tick labels = {5},
							ymajorgrids = true,
							xlabel = Number of positions, %ylabel = Time (sec),
							y label style = {at = {(axis description cs:0.02,1)}}
						]

						\addplot [blue!60!black, solid, line width = 1pt, mark = pentagon, mark options = solid, mark size = 1] table [x index = 0, y index = 1]{big2linearprio.txt};

 						\addplot [black!50!black, solid, line width = 1pt, mark = o, mark	options = solid, mark size = 1	] table [x index = 0, y index = 3]{big2linearprio.txt};
 						\addplot [green!50!black, solid, line width = 1pt, mark = o, mark	options = solid, mark size = 1	] table [x index = 0, y index = 4]{big2linearprio.txt};
 						\addplot [green, solid, line width = 1pt, mark = o, mark	options = solid, mark size = 1	] table [x index = 0, y index = 5]{big2linearprio.txt};

            \end{axis}
            \begin{axis}
						[
							at = (lowerpart.north), anchor = south,
							width = 0.60\textwidth, height = 0.25\textwidth,
							xmin = 1, xmax = 10, ymin = 5, ymax = 1800,
              xtick={1,2,4,6,8,10},
              xticklabels = {$10^4$,$2\cdot 10^4$,$4\cdot 10^4$,$6\cdot 10^4$,$8\cdot 10^4$,$10^5$},
							ytick = {100, 650, ..., 1750},
							scale only axis,
							axis x line* = top,
							x axis line style = -, y axis line style = -,
							minor y tick num = 1,
              ylabel = Time (sec),
              y label style={at={(0.02,0.38)}},
							ymajorgrids = true,
							xticklabel = \empty,
							legend style={at={(0.98,0.45)}},
							legend entries = {HPP, Par, RPP, PP, ZLK}
						]

  						\addplot [blue!60!black, solid, line width = 1pt, mark = pentagon, mark options = solid, mark size = 1] table [x index = 0, y index = 1,  skip coords between index={0}{9}] {big2linearprio.txt};

              \addplot [red!50!black, solid, line width = 1pt, mark = o, mark options = solid, mark size = 1] table [x index = 0, y index = 2] %,  skip coords between index={0}{9}]
							{big2linearprio.txt};

              \addplot [green!50!black, solid, line width = 0.75pt, mark = diamond, mark options = solid, mark size = 1] table [x index = 0, y index = 3,  skip coords between index={0}{9}]{big2linearprio.txt};

              \addplot [black, solid, line width = 0.75pt, mark = +, mark	options = solid, mark size = 1] table [x index = 0, y index = 4,  skip coords between index={0}{9}]{big2linearprio.txt};

              \addplot [green, solid, line width = 0.75pt, mark = +, mark	options = solid, mark size = 1] table [x index = 0, y index = 5,  skip coords between index={0}{9}] {big2linearprio.txt};

					\end{axis}
				\end{tikzpicture}
				\begin{tikzpicture}[scale = 0.7, every node/.style = {scale = 0.675}]
					\begin{axis}
						[
							name = lowerpart,
							width = 0.60\textwidth, height = 0.10\textwidth,
							xmin = 1, xmax = 10, ymin = 0, ymax = 5,
							scale only axis,
							axis x line* = bottom,
							x axis line style = -, y axis line style = -,
              xtick={1,2,4,6,8,10},
              xticklabels = {$10^4$,$2\cdot 10^4$,$4\cdot 10^4$,$6\cdot 10^4$,$8\cdot 10^4$,$10^5$},
							ytick = {0, 1, ..., 5},
							extra y ticks = 5,
							extra y tick style = {grid style = {black, dashed}},
							extra y tick labels = {5},
							ymajorgrids = true,
							xlabel = Number of positions, %ylabel = Time (sec),
							y label style = {at = {(axis description cs:0.02,1)}}
						]

						\addplot [blue!60!black, solid, line width = 1pt, mark = pentagon,
						  mark options = solid, mark size = 1] table [x index = 0, y index = 1]
							{big21000prio.txt};

              \addplot [green!50!black, , solid, line width = 1pt, mark = diamond,
              mark options = solid, mark size = 1] table [x index = 0, y index = 3]
							{big21000prio.txt};

						\addplot [green, solid, line width = 1pt, mark = +, mark
							options = solid, mark size = 1] table [x index = 0, y index = 5]
							{big21000prio.txt};

              \addplot [black, solid, line width = 1pt, mark = +, mark
							options = solid, mark size = 1] table [x index = 0, y index = 4]
							{big21000prio.txt};

            \end{axis}
					\begin{axis}
						[
							at = (lowerpart.north), anchor = south,
							width = 0.60\textwidth, height = 0.25\textwidth,
							xmin = 1, xmax = 10, ymin = 4, ymax = 1850,
							ytick = {100, 650, ..., 1750},
              xtick={1,2,4,6,8,10},
              xticklabels = {$10^4$,$2\cdot 10^4$,$4\cdot 10^4$,$6\cdot 10^4$,$8\cdot 10^4$,$10^5$},
							scale only axis,
							axis x line* = top,
							x axis line style = -, y axis line style = -,
							minor y tick num = 1,
							ymajorgrids = true,
              ylabel = Time (sec),
              y label style={at={(0.02,0.38)}},
							xticklabel = \empty,
							legend style={at={(0.98,0.45)}},
							legend entries = {HPP,Par, RPP, PP, ZLK}
						]

						\addplot [blue!60!black, solid, line width = 1pt, mark = pentagon,
						  mark options = solid, mark size = 1] table [x index = 0, y index = 1,  skip coords between index={0}{9}]
							{big21000prio.txt};

						\addplot [red!60!black, solid, line width = 1pt, mark = o, mark
							options = solid, mark size = 1
							] table [x index = 0, y index = 2] {big21000prio.txt};

              \addplot [green!50!black, , solid, line width = 1pt, mark = diamond,
              mark options = solid, mark size = 1] table [x index = 0, y index = 3,  skip coords between index={0}{9}]
							{big21000prio.txt};

                        \addplot [black, solid, line width = 1pt, mark = +, mark
							options = solid, mark size = 1] table [x index = 0, y index = 4,  skip coords between index={0}{9}]
							{big21000prio.txt};
							
						\addplot [green, solid, line width = 1pt, mark = +, mark
							options = solid, mark size = 1] table [x index = 0, y index = 5,  skip coords between index={0}{9}]
							{big21000prio.txt};

					\end{axis}
				\end{tikzpicture}
				\vspace{-0.50em}
				\caption{\label{fig:bigper} \small  Results on random games with a linear (left) and fixed (right) number of priorities.}
  \end{center}
  \vspace{-2.5em}
  \end{figure}
   }

  }

%% file: Tables.tex
\usepackage{multirow}

% ...

\AfterEndPreamble
  {

  \newcommand{\tabper}
    {
          \begin{tabular}{|l|l||c|c|c||c|c|c|>{\color{blue}} c|}
            \hline \multicolumn{2}{|c||}{\ } & \multicolumn{3}{c||}{\ } & \multicolumn{4}{c|}{\ } \\[-0.80em]
             \multicolumn{2}{|c||}{\ }  & \multicolumn{3}{c||}{Exponential} & \multicolumn{4}{c|}{Quasi Polynomial}  \\
            \hline  \multicolumn{2}{|c||}{\ } & & & & & & &\\[-.9em]
            \multicolumn{2}{|c||}{Benchmarks}   & \multicolumn{1}{c|}{ZLK} & \multicolumn{1}{c|}{NPP} & \multicolumn{1}{c||}{RPP} & \multicolumn{1}{c|}{SSPM} & \multicolumn{1}{c|}{QPT} & \multicolumn{1}{c|}{$\ $ZLKQ$\ $} &  \multicolumn{1}{c|}{\textbf{\color{blue}HPP}}  \\
            \hline & &  && & & & & \\[-0.80em]
                     & \text{\scriptsize Tot. Time}      & 27.16 &  20.12   &   53.66    & {\color{red} $\ >$849.01$\ $}  &  {\color{red}$>$1259.84}    & 42.38  & \textbf{64.21}\\
            Model-Ch.  & \text{\scriptsize Avg. Time}    & 0.08 & 0.06    &    0.17    & {\color{red}$>$2.71}  & {\color{red} $>$4.02}    & 0.13 & \textbf{0.2}  \\
                     & \text{\scriptsize  TimeOut}    & \multicolumn{1}{c|}{0\%} &  \multicolumn{1}{c|}{0\%} &  \multicolumn{1}{c||}{0\%} & \multicolumn{1}{c|}{\color{red}22\%}  & \multicolumn{1}{c|}{\color{red}36.4\%}     & \multicolumn{1}{c|}{0\%} & \multicolumn{1}{c|}{\textbf{\color{blue}0\%}} \\
            \hline & & & & & & & &\\[-.9em]
                     & \text{\scriptsize Tot. Time}      & 202.95 & 137.92    &    242.32       & {\color{red}$\ >$2681.33$\ $}  & {\color{red}$>$3139.85}    & 208.3  & $\ $\textbf{280.81}$\ $\\
            Equiv. Ch.  & \text{\scriptsize Avg. Time}    & 0.94 & 0.63       &    1.12         &  {\color{red} $>$12.41}  & {\color{red}$>$14.53}   &   0.93  & \textbf{1.3}  \\
                     & \text{\scriptsize  TimeOut}    & \multicolumn{1}{c|}{0\%} &  \multicolumn{1}{c|}{0\%} &  \multicolumn{1}{c||}{0\%} & \multicolumn{1}{c|}{\color{red}28.2\%}  & \multicolumn{1}{c|}{\color{red}27.3\%}     & \multicolumn{1}{c|}{0\%} & \multicolumn{1}{c|}{\textbf{\color{blue}0\%}} \\
            \hline  & & & & &  && &\\[-.9em]
                     & \text{\scriptsize Tot. Time}      & 13.20 &  11.75   &   13.27        & {\color{red}$>$360.8}  & {\color{red}$>$853.64}    & 66.85  & \textbf{14.27}\\
            Decision Prb.  & \text{\scriptsize Avg. Time} & 0.07 &  0.06    &   0.07         & {\color{red}$>$1.87}  & {\color{red}$>$4.44}  & 0.35  & \textbf{0.07}  \\
                     & \text{\scriptsize  TimeOut}    & \multicolumn{1}{c|}{0\%} &  \multicolumn{1}{c|}{0\%} &  \multicolumn{1}{c||}{0\%} & \multicolumn{1}{c|}{\color{red}11\%}  & \multicolumn{1}{c|}{\color{red}26.5\%}     & \multicolumn{1}{c|}{0\%} & \multicolumn{1}{c|}{\textbf{\color{blue}0\%}} \\
            \hline  & & & & & & & & \\[-0.80em]
                     & \text{\scriptsize Tot. Time}      & 1.54  &  2.21    &   2.89         & {\color{red}$\ >$3615.22$\ $}  & {\color{red}$\ >$4069.12$\ $}   &  8.62  & \textbf{4.04}\\
            PGSolver  & \text{\scriptsize Avg. Time}    &  0.005 &  0.007   &  0.009         & {\color{red}$>$12.42}  &  {\color{red}$>$13.98}     &  0.03 & \textbf{0.01}  \\
                     & \text{\scriptsize  TimeOut}    & \multicolumn{1}{c|}{0\%} &  \multicolumn{1}{c|}{0\%} &  \multicolumn{1}{c||}{0\%} & \multicolumn{1}{c|}{\color{red}78\%}  & \multicolumn{1}{c|}{\color{red}92.4\%}     & \multicolumn{1}{c|}{0\%}  & \multicolumn{1}{c|}{\textbf{\color{blue}0\%}} \\
            \hline
          \end{tabular}
    }
    
    \newcommand{\tabwc}
    {
      \begin{tabular}{|l|l||c|c||c|c||c|c|c|}
            \hline \multicolumn{2}{|c||}{\ } & \multicolumn{2}{c||}{\ } & \multicolumn{2}{c||}{\ } & \multicolumn{2}{c|}{\ } \\[-0.80em]
             \multicolumn{2}{|c||}{Worst Case}  & \multicolumn{2}{c||}{NPP} & \multicolumn{2}{c||}{u-ZLK} & \multicolumn{2}{c|}{HPP}  \\
            \hline  \multicolumn{2}{|c||}{\ } & & & & & &\\[-.9em]
            \multicolumn{1}{c|}{Index} & \multicolumn{1}{c|}{Positions} & \multicolumn{1}{c|}{Time} & \multicolumn{1}{c||}{Iterations} & \multicolumn{1}{c|}{Time} & \multicolumn{1}{c||}{Iterations} & \multicolumn{1}{c|}{Time} & \multicolumn{1}{c|}{Iterations}\\
            \hline
            10 & 88 & 0.00 & 11267 & 0.00 & 1295 & 0.00 & 69\\
            15 & 170 & 0.18 & 524294 & 0.00 & 10175 & 0.00 & 177\\
            20 & 278 & 7.70 & 22020104 & 0.03 & 133631 & 0.00 & 339\\
            25 & 410 & - & - & 0.17 & 796671 & 0.00 & 879\\
            30 & 568 & - & - & 1.99 & 8863743 & 0.00 & 1217\\
            35 & 750 & - & - & 11.25 & 47120383 & 0.01 & 2125\\
            40 & 958 & - & - & - & - & 0.02 & 2952\\
            \hline
      \end{tabular}
    }

  }

%% file: Algorithms.tex
\AfterEndPreamble{

\newcommand{\algsolexpl}{
\begin{algorithm}[H]
  \caption{\label{alg:expsol} \RPP Solver}
  \SetInd{0.25em}{0.5em}
  \Function{$\solFun(\sttElm \colon \SttSet[\MSym]) \colon \RegSet$}
    {
    \nl \If{$\prtElm[\sttElm] \neq \bot$}
      {
      \nl $\RFun[\sttElm] \gets \atrFun[ {\GameName[\sttElm]}
          ][\alpha_{\sttElm}][ {\RFun[\sttElm]} ]$ \;
      \nl \eIf{$\sttElm$ is open}
        {
        \nl $\regFun[\sttElm] \gets \solFun(\NxtFun(\sttElm))$ \;
        \nl \eIf{$\sttElm$ is open}
          {
          \nl $\MaxFun(\sttElm)$ \;
          }
          {
          \nl $\PrmFun(\sttElm)$ \;
          }
        }
        {
        \nl $\PrmFun(\sttElm)$ \;
        }
      \nl $\regFun[\sttElm] \gets \solFun(\sttElm)$ \;
      }
    \nl \Return $\regFun[\sttElm]$ \;
    }
\end{algorithm}
}

\newcommand{\algsolexp}{
\begin{algorithm}[H]
  \caption{\RPP Solver}
  \SetInd{0.25em}{0.5em}
  \Function{$\solFun(\sttElm \colon \SttSet[\MSym]) \colon \RegSet$}
    {
    \nl \If{$\prtElm[\sttElm] \neq \bot$}
      {
      \nl $\RFun[\sttElm] \gets \atrFun[ {\GameName[\sttElm]}
          ][\alpha_{\sttElm}][ {\RFun[\sttElm]} ]$ \;
      \nl \eIf{$\sttElm$ is open}
        {
        \nl $\regFun[\sttElm] \gets \solFun(\NxtFun(\sttElm))$ \;
        \nl \eIf{$\sttElm$ is open}
          {
          \nl $\MaxFun(\sttElm)$ \;
          }
          {
          \nl $\PrmFun(\sttElm)$ \;
          }
        }
        {
        \nl $\PrmFun(\sttElm)$ \;
        }
      \nl $\regFun[\sttElm] \gets \solFun(\sttElm)$ \;
      }
    \nl \Return $\regFun[\sttElm]$ \;
    }
\end{algorithm}
}

\newcommand{\algsolexpmacro}{
\NoCaptionOfAlgo
\begin{algorithm}[H]
  \SetInd{0.25em}{0.5em}
  \caption{Auxiliary Functions / Procedures}
  \vspace{0.125em}
  \Function{$\NxtFun(\sttElm \colon \SttSet[\SSym]) \colon \SttSet[\MSym]$}
    {
    \nl $\qprtElm \gets \max(\rng{\regFun[\sttElm]^{(< \prtElm[\sttElm])}})$ \;
    \nl \Return $(\regFun[\sttElm], \qprtElm)$ \;
    }
  \vspace{0.265em}
  \setcounter{AlgoLine}{0}
  \Procedure{$\MaxFun(\sttElm \colon \SttSet)$}
    {
    \nl \ForEach{$\alpha \in \SetB$}
      {
      \nl $\XSet \gets \atrFun[][\alpha][ {\HSet[\sttElm][\alpha] \setminus
          \LSet[\sttElm], \LSet[\sttElm]} ]$ \;
      \nl $\qprtElm \gets \min(\rng{ \regFun[\sttElm] \rst\,
          (\HSet[\sttElm][\alpha] \setminus \LSet[\sttElm]) })$ \;
      \nl $\regFun[\sttElm] \gets {\regFun[\sttElm]}[\XSet \mapsto
          \qprtElm]$ \;
      }
    \nl $\regFun[\sttElm] \gets {\regFun[\sttElm]}[\posElm \in \LSet[\sttElm]
        \mapsto \prtFun(\posElm)]$ \;
    }
  \vspace{0.265em}
  \setcounter{AlgoLine}{0}
  \Procedure{$\PrmFun(\sttElm \colon \SttSet[\PSym])$}
    {
    \nl $\qprtElm \gets \bepFun[][ {\dual{\alpha}[\sttElm]} ][{\RFun[\sttElm]},
        {\regFun[\sttElm]}]$ \;
    \nl $\regFun[\sttElm] \gets {\regFun[\sttElm]}[{\RFun[\sttElm]} \mapsto
        {\qprtElm}]$ \;
    }
  \vspace{0.125em}
\end{algorithm}
}

\newcommand{\algsolqsipoll}{
\begin{algorithm}[H]
  \caption{\label{alg:hybsol} \HPP Solver}
  \SetInd{0.25em}{0.5em}
  \vspace{0.12em}
  \Function{$\solFun(\sttElm \colon \SttSet[\MSym]) \colon \RegSet
    \times \PrmSet$}
    {
    \nl \eIf{$\prtElm[\sttElm] \!=\! \bot \vee \bElm[0\sttElm] \!=\! 0 \vee
        \bElm[1\sttElm] \!=\! 0$}
      {
      \nl \Return $(\regFun[\sttElm], \undFun[\sttElm])$ \;
      }
      {
      \nl $\hsolFun(\sttElm)$ \;
      \nl $\der{\sttElm} \gets \sttElm$ \;
      \nl $(\regFun[\sttElm], \undFun[\sttElm]) \gets
          \solFun(\NxtFun(\sttElm))$ \;
      \nl \eIf{$\sttElm$ is open}
        {
        \nl $\MaxFun(\sttElm)$ \;
        }
        {
        \nl $\PrmFun(\sttElm)$ \;
        }
      \nl \lIf{$\sttElm \prec \der{\sttElm}$\,}
        {
        $\hsolFun(\sttElm)$
        }
      \nl \Return $\UndFun(\sttElm)$ \;
      }
    }
  \vspace{0.1175em}
\end{algorithm}
}

\newcommand{\algsolqsipol}{
\begin{algorithm}[H]
  \caption{\HPP Solver}
  \SetInd{0.25em}{0.5em}
  \vspace{0.12em}
  \Function{$\solFun(\sttElm \colon \SttSet[\MSym]) \colon \RegSet
    \times \PrmSet$}
    {
    \nl \eIf{$\prtElm[\sttElm] \!=\! \bot \vee \bElm[0\sttElm] \!=\! 0 \vee
        \bElm[1\sttElm] \!=\! 0$}
      {
      \nl \Return $(\regFun[\sttElm], \undFun[\sttElm])$ \;
      }
      {
      \nl $\hsolFun(\sttElm)$ \;
      \nl $\der{\sttElm} \gets \sttElm$ \;
      \nl $(\regFun[\sttElm], \undFun[\sttElm]) \gets
          \solFun(\NxtFun(\sttElm))$ \;
      \nl \eIf{$\sttElm$ is open}
        {
        \nl $\MaxFun(\sttElm)$ \;
        }
        {
        \nl $\PrmFun(\sttElm)$ \;
        }
      \nl \lIf{$\sttElm \prec \der{\sttElm}$\,}
        {
        $\hsolFun(\sttElm)$
        }
      \nl \Return $\UndFun(\sttElm)$ \;
      }
    }
  \vspace{0.1175em}
\end{algorithm}
}

\newcommand{\algsolqsipolhalf}{
\begin{algorithm}[H]
  \caption{Half-Solver{\protect\vphantom{\HPP}}}
  \SetInd{0.25em}{0.5em}
  \Procedure{$\hsolFun(\sttElm \colon \SttSet[\MSym])$}
    {
    \nl \Repeat{$\sttElm \not\prec \der{\sttElm}$}
      {
      \nl $\RFun[\sttElm] \gets \atrFun[ {\GameName[\sttElm]}
          ][\alpha_{\sttElm}][ {\RFun[\sttElm]} ]$ \;
      \nl $\der{\sttElm} \gets \sttElm$ \;
      \nl \eIf{$\sttElm$ is open}
        {
        \nl $(\regFun[\sttElm], \undFun[\sttElm]) \gets
            \solFun(\HalfFun(\sttElm))$ \;
        \nl \eIf{$\sttElm$ is open}
          {
          \nl $\MaxFun(\sttElm)$ \;
          }
          {
          \nl $\PrmFun(\sttElm)$ \;
          }
        }
        {
        \nl $\PrmFun(\sttElm)$ \;
        }
      }
    }
\end{algorithm}
}

\newcommand{\algsolqsipolmacroi}{
\NoCaptionOfAlgo
\begin{algorithm}[H]
  \caption{Auxiliary Functions / Procedures I\!\!}
  \SetInd{0.25em}{0.5em}
  \vspace{1em}
  \Function{$\NxtFun(\sttElm \colon \SttSet[\SSym]) \colon \SttSet[\MSym]$}
    {
    \nl $\qprtElm \gets \max(\rng{\regFun[\sttElm]^{(< \prtElm[\sttElm])}})$ \;
    \nl \Return $((\regFun[\sttElm], \qprtElm), (\undFun[\sttElm],
        \prtElm[\sttElm], \bElm[0\sttElm], \bElm[1\sttElm]))\!\!$ \;
    }
  \vspace{1.2em}
  \setcounter{AlgoLine}{0}
  \Function{$\HalfFun(\sttElm \colon \SttSet[\SSym]) \colon
    \SttSet[\MSym]$}
    {
    \nl $(\bElm[0\sttElm], \bElm[1\sttElm]) \gets
        (\floor{\frac{\bElm[0\sttElm]}{1 + \alpha_{\sttElm}}},
        \floor{\frac{\bElm[1\sttElm]}{2 - \alpha_{\sttElm}}})$ \;
    \nl \Return $\NxtFun(\sttElm)$ \;
    }
  \vspace{1.2em}
  \setcounter{AlgoLine}{0}
  \Function{$\UndFun(\sttElm \colon \SttSet[\SSym]) \colon \RegSet \!\times\!
    \PrmSet$\!\!\!\!\!\!\!}
    {
      \nl \eIf{$\cprtElm[\sttElm] \equiv_{2} \alpha_{\sttElm}$}
        {
        \nl $\undFun[\sttElm] \gets {\undFun[\sttElm]}[{ \USet[\sttElm] \mapsto
            \cprtElm[\sttElm] }]$ \;
        }
        {
        \nl $\regFun[\sttElm] \gets {\regFun[\sttElm]^{(\geq
            \cprtElm[\sttElm])}}[\posElm \!\in\! \USet[\sttElm] \mapsto
            \prtFun(\posElm)]\!\!$ \;
        \nl $\undFun[\sttElm] \gets \undFun[\sttElm]^{(\geq
            \cprtElm[\sttElm])}{[\LSet[\sttElm] \mapsto \cprtElm[\sttElm]]}$ \;
        }
    \nl \Return $(\regFun[\sttElm], \undFun[\sttElm])$ \;
    }
  \vspace{1em}
\end{algorithm}
}

\newcommand{\algsolqsipolmacroii}{
\NoCaptionOfAlgo
\setlength{\algomargin}{0.5em}
\begin{algorithm}[H]
  \caption{Auxiliary Functions / Procedures II}
  \SetInd{0.25em}{0.5em}
  \setcounter{AlgoLine}{0}
  \Procedure{$\PrmFun(\sttElm \colon \SttSet[\PSym])$}
    {
    \nl $(\prtElm[\regFun], \prtElm[\undFun]) \!\gets\!
        ({\bepFun[][ {\dual{\alpha}[\sttElm]} ][{\RFun[\sttElm]},
        {\regFun[\sttElm]}]},
        {\bepFun[][ {\dual{\alpha}[\sttElm]} ][{\RFun[\sttElm]},
        {\undFun[\sttElm]}]})\!\!\!\!\!\!\!$ \;
    \nl \eIf{$\prtElm[\regFun] \leq \prtElm[\undFun]$}
      {
      \nl $\regFun[\sttElm] \gets {\regFun[\sttElm]}[{\RFun[\sttElm]} \mapsto
          {\prtElm[\regFun]}]$ \;
      }
      {
      \nl $(\regFun[\sttElm], \undFun[\sttElm]) \gets (\regFun[\sttElm]
          \!\setminus\! \RSet[\sttElm], {\undFun[\sttElm]}[{\RSet[\sttElm]}
          \mapsto {\prtElm[\undFun]}])$ \;
      }
    }
  \setcounter{AlgoLine}{0}
  \Procedure{$\MaxFun(\sttElm \colon \SttSet)$}
    {
    \nl $\ZSet \gets \RSet[\sttElm]$ \;
    \nl \ForEach{$\alpha \in \SetB$}
      {
      \nl $\XSet \gets \atrFun[][\alpha][ {\HSet[\sttElm][\alpha] \setminus
          \LSet[\sttElm], \LSet[\sttElm] \cup \USet[\sttElm]} ]$ \;

      \nl $\qprtElm \gets \min(\rng{ (\regFun[\sttElm] \cup \undFun[\sttElm])
          \rst\, (\HSet[\sttElm][\alpha] \setminus \LSet[\sttElm]) })\!\!$ \;

      \nl \eIf{$\qprtElm \equiv_{2} \alpha$}
        {
        \nl $\regFun[\sttElm] \gets {\regFun[\sttElm]}[\XSet \mapsto \qprtElm]$
            \;
        \nl $\undFun[\sttElm] \!\gets \undFun[\sttElm] \setminus \XSet$ \;
        }
        {
        \nl $\regFun[\sttElm] \gets {\regFun[\sttElm]} \setminus \XSet$ \;
        \nl $\undFun[\sttElm] \gets {\undFun[\sttElm]}[\XSet \mapsto \qprtElm]$
            \;
        }
      }
    \nl \lIf{$\ZSet \neq \RSet[\sttElm]$}
      {
      \!$\regFun[\sttElm] \!\gets\! {\regFun[\sttElm]}[\posElm \!\in\!
      \LSet[\sttElm] \!\mapsto\! \prtFun(\posElm)]$\!
      }
    }
\end{algorithm}
}

\newcommand{\algsolparl}{
\begin{algorithm}[H]
  \caption{\label{alg:parsol} Parys Solver}
  \SetInd{0.25em}{0.5em}
  \vspace{0.12em}
  \Function{$\solFun(\sttElm \colon \SttSet) \colon 2^{\PosSet}$}
    {
    \nl \eIf{$\prtElm[\sttElm] \!=\! \bot \vee \bElm[0\sttElm] \!=\! 0 \vee
        \bElm[1\sttElm] \!=\! 0$}
      {
      \nl \Return $(\regFun[\sttElm],\undFun[\sttElm])$ \;
      }
      {
      \nl $\hsolFun(\sttElm)$ \;
      \nl $\der{\sttElm} \gets \sttElm$ \;
      \nl $(\regFun[\sttElm], \undFun[\sttElm]) \gets
          \solFun(\NxtFun(\sttElm))$ \;
      \nl $\MaxFun(\sttElm)$ \;
      \nl \lIf{$\sttElm \prec \der{\sttElm}$\,}
        {
        $\hsolFun(\sttElm)$
        }
      \nl \Return $\UndFun(\sttElm)$ \;
      }
    }
  \vspace{0.1175em}
\end{algorithm}
}

\newcommand{\algsolpar}{
\begin{algorithm}[H]
  \caption{Parys Solver}
  \SetInd{0.25em}{0.5em}
  \vspace{0.12em}
  \Function{$\solFun(\sttElm \colon \SttSet) \colon 2^{\PosSet}$}
    {
    \nl \eIf{$\prtElm[\sttElm] \!=\! \bot \vee \bElm[0\sttElm] \!=\! 0 \vee
        \bElm[1\sttElm] \!=\! 0$}
      {
      \nl \Return $(\regFun[\sttElm],\undFun[\sttElm])$ \;
      }
      {
      \nl $\hsolFun(\sttElm)$ \;
      \nl $\der{\sttElm} \gets \sttElm$ \;
      \nl $(\regFun[\sttElm], \undFun[\sttElm]) \gets
          \solFun(\NxtFun(\sttElm))$ \;
      \nl $\MaxFun(\sttElm)$ \;
      \nl \lIf{$\sttElm \prec \der{\sttElm}$\,}
        {
        $\hsolFun(\sttElm)$
        }
      \nl \Return $\UndFun(\sttElm)$ \;
      }
    }
  \vspace{0.1175em}
\end{algorithm}
}

\newcommand{\algsolparhalf}{
\begin{algorithm}[H]
  \caption{\label{alg:parsolhalf}Half-Solver{\protect\vphantom{$Par$}}}
  \SetInd{0.25em}{0.5em}
  \Procedure{$\hsolFun(\sttElm \colon \SttSet)$}
    {
    \nl \Repeat{$\sttElm \not\prec \der{\sttElm}$}
      {
      \nl $\RFun[\sttElm] \gets \atrFun[ {\GameName[\sttElm]}
          ][\alpha_{\sttElm}][ {\RFun[\sttElm]} ]$ \;
      \nl $\der{\sttElm} \gets \sttElm$ \;
        \nl $(\regFun[\sttElm], \undFun[\sttElm]) \gets
            \solFun(\HalfFun(\sttElm))$ \;
        \nl $\MaxFun(\sttElm)$ \;
      }
    }
\end{algorithm}
}

\newcommand{\algsolparmacroi}{
\NoCaptionOfAlgo
\begin{algorithm}[H]
  \caption{Auxiliary Functions / Procedures I\!\!}
  \SetInd{0.25em}{0.5em}
  \Function{$\NxtFun(\sttElm \colon \SttSet) \colon \SttSet$}
    {
    \nl \Return $((\regFun[\sttElm], \prtElm[\sttElm]\!-\!1), (\undFun[\sttElm], \prtElm[\sttElm], \bElm[0\sttElm], \bElm[1\sttElm]))$\!\! \;
    }
  \setcounter{AlgoLine}{0}
  \Function{$\HalfFun(\sttElm \colon \SttSet) \colon \SttSet$}
    {
    \nl $(\bElm[0\sttElm], \bElm[1\sttElm]) \gets
        (\floor{\frac{\bElm[0\sttElm]}{1 + \alpha_{\sttElm}}},
        \floor{\frac{\bElm[1\sttElm]}{2 - \alpha_{\sttElm}}})$ \;
    \nl \Return $\NxtFun(\sttElm)$ \;
    }
\end{algorithm}
}

\newcommand{\algsolparmacroii}{
\NoCaptionOfAlgo
\setlength{\algomargin}{0.5em}
\begin{algorithm}[H]
  \caption{Auxiliary Functions / Procedures II\!\!\!\!\!\!\!\!\!\!\!}
  \SetInd{0.25em}{0.5em}
  \setcounter{AlgoLine}{0}
  \Procedure{$\MaxFun(\sttElm \colon \SttSet)$}
    {
      \nl $\XSet \gets \atrFun[][
          {\dual{\alpha}[\sttElm]} ][{ \HSet[\sttElm][ {\dual{\alpha}[\sttElm]}], \LSet[\sttElm]}]$ \;
      \nl $\regFun[\sttElm] \gets \regFun[\sttElm] \setminus \XSet$ \;
      \nl $\undFun[\sttElm] \!\gets {\undFun[\sttElm]}[\XSet \mapsto \prtElm[\sttElm]]$ \;
    }
  \vspace{0.5em}
  \setcounter{AlgoLine}{0}
  \Function{$\UndFun(\sttElm \colon \SttSet) : \RegSet \!\times\! \PrmSet$}
    {
      \nl $\regFun[\sttElm] \gets {\regFun[\sttElm]^{(>\prtElm[\sttElm])}}[\posElm \!\in\! \USet[\sttElm] \mapsto
            \prtFun(\posElm)]\!\!$ \;
      \nl $\undFun[\sttElm] \gets \undFun[\sttElm]{[\LSet[\sttElm] \mapsto \prtElm[\sttElm]+1]}$ \;
      \nl \Return $(\regFun[\sttElm], \undFun[\sttElm])$ \;
    }
\end{algorithm}
}

}

%% file: Abstract.tex
% Begin of file Abstract.tex

\begin{abstract}

  We develop an algorithm that combines the advantages of priority promotion - one of the leading approaches to solving large parity games in practice - with the quasi-polynomial time guarantees offered by Parys' algorithm. Hybridising these algorithms sounds both natural and difficult, as they both generalise the classic recursive algorithm in different ways that appear to be irreconcilable: while the promotion transcends the call structure, the guarantees change on each level. We show that an interface that respects both is not only effective, but also efficient.

%  \keywords{Parity Games \and Priority Promotion \and Quasi-Polynomial Time}

\end{abstract}

% End of file Abstract.tex

%% file: Introduction.tex
% Begin of file Introduction.tex

\section{Introduction}

Parity games have many applications in model checking~\cite{Koz83,EJS93,%
deAlfaro+Henzinger+Majumdar/01/Control,%
AHK02,Wil01,KV97} and synthesis~\cite{Wil01,Koz83,Var98,SF06,Pit06,%
Schewe/09/determinise,Schewe+Finkbeiner/06/Asynchronous}.
In particular, modal and alternating-time $\mu$-calculus model
checking~\cite{Wil01,AHK02}, synthesis~\cite{Schewe+Finkbeiner/06/Asynchronous,%
Pit06,%
Schewe/09/determinise} and satisfiability checking~\cite{Wil01,Koz83,Var98,SF06}
for reactive systems, module checking~\cite{KV97}, and \ATLS model
checking~\cite{deAlfaro+Henzinger+Majumdar/01/Control,%
AHK02} can be reduced to solving parity games.
This relevance of parity games led to a series of different approaches to
solving them
\cite{McN93,EL86,%
Ludwig/95/random,Puri/95/simprove,%
ZP96,%
Browne-all/97/fixedpoint,%
Zie98,Jur98,Jur00,VJ00,%
Obdrzalek/03/TreeWidth,Lange/05/ParitySAT,%
BDHK06,BV07,JPZ08,FS13,BDM18,CJKLS17,JL17,Leh18,FJKSSW19,Par19,LSW19,Par20,%
DJT20,LPSW21,Dij18a,LBD20}.

The research falls into two categories: on the one hand to develop \emph{fast solvers}; 
on the other hand to
determine the \emph{complexity of parity games} or to find algorithms with a
good \emph{worst-case complexity}.
With its practical motivation, one of the leading algorithms most efficient for solving parity games is currently \emph{priority promotion}
techniques~\cite{BDM16b,BDM18,BDM18a}, a refinement of the classic \emph{recursive
algorithm}~\cite{McN93,EL86,Zie98} that follows the iterated fixed-point
structure induced by the parity condition.
The complexity of solving parity games is still an open problem.
Parity games are \emph{memoryless determined}~\cite{EJS93,Bjorklund2004365},
which implies that nondeterministic algorithms can determine winning regions and
strategies for both players.
Due to their symmetry, they are therefore in
\NPTime~$\cap$~\CoNPTime~\cite{EJS93}, and by reduction to payoff
games~\cite{ZP96}, in \UPTime~$\cap$~\CoUPTime~\cite{Jur98}.
While determining their membership in \PTime continues to be a major challenge,
one of the most celebrated results in recent years has been the landmark result
of Calude \etal~\cite{CJKLS17}, which established that parity games can be
solved in quasi-polynomial time (QP).
This was a major step from former deterministic algorithms, which were (at
least) exponential in the number of priorities~\cite{McN93,EL86,ZP96,%
Browne-all/97/fixedpoint,%
Zie98,Jur00,BV07,Sch08a,Sch17,BDM18} ($n^{O(c)}$), or in the square-root of the
number of game positions~\cite{Ludwig/95/random,%
JPZ08,BV07} (approximately $n^{O(\sqrt{n})}$).
The breakthrough of Calude \etal~\cite{CJKLS17} has triggered a new line of
research into QP algorithms,
including~\cite{JL17,Leh18,FJKSSW19,Par19,LSW19,Par20,DJT20}.

Algorithms that are good in practice do not tend to display their worst-case
behaviour, except for in carefully designed hostile examples.
This holds in particular for strategy improvement
algorithms~\cite{Ludwig/95/random,Puri/95/simprove,%
VJ00,BV07,Sch08a,Fea10a,STV15}, which were considered candidates for tractable
algorithms until they were shown to be exponential by Friedman's delicate lower
bound constructions~\cite{Fri11b,%
Friedmann/11/Zadeh,%
Fri13} (with the notable exception of the symmetric approach from~\cite{STV15},
for which no hard families are known).
But while it is easier to design hard classes for recursive~\cite{Fri11a,BDM17}
and priority-promotion algorithms~\cite{BDM18a}, these classes are still not
relevant in practice.
However, with a host of QP algorithms at hand, an upgrade of priority promotion
that offers QP lower bounds without undue compromise on efficiency will be an
attractive challenge that combines the best of both worlds.

Interestingly, Parys' algorithm~\cite{Par19} and variations
thereof~\cite{LSW19,LPSW21}, which, like priority promotion techniques, adjust the
classic recursive algorithm~\cite{McN93,EL86,Zie98}, are relatively fast among
the QP algorithms, where~\cite{Par19} has the edge on benchmarks,
while~\cite{LSW19} has the edge on theoretical guarantees.
On first glance, this seems to invite synthesising one of these algorithms with
priority promotion.
On second glance, the prospect of this synthesis seems less promising.
Priority promotion techniques~\cite{BDM16b,BDM18,BDM18a} achieve their advancement over
the previously leading recursive algorithms~\cite{McN93,EL86,Zie98} by globally
bypassing the call structure through temporally increasing the priority of a
position.
Parys' approach, on the other hand, locally creates sets with guarantees with
quickly falling strength along the recursive call structure, where subgames are
split into areas that contain all $0$-dominions with size up to a bound $b_0$
and all $1$-dominions of size up to a bound $b_1$; one of these bounds is halved
in each call until the guarantees are trivial.
\Primafacie, it seems clear that such guarantees are ill suited for a promotion
across the call structure.
We did, however, find that, when one shifts the view on the essence of a
promotion from creating quasi-dominions to creating regions and promoting them
to the lowest level where they are no longer dominions, this allows for a
concurrent treatment of sets with bounded guarantees (the Parysian flair of our
hybrid algorithm) and with unbounded guarantees (the Priority Promotion core of
our algorithm).
While the integration of these seemingly antagonistic concepts is intricate, it
provides an efficient bridge between the behaviour and the data structure
of~\cite{BDM18} and~\cite{Par19}: the resulting algorithm guarantees a
quasi-polynomial running time, and offers excellent practical behaviour on the
benchmarks we have tested it against.

% End of file Introduction.tex

%% file: Preliminaries.tex
% Begin of file Preliminaries.tex

\section{Preliminaries}
\label{sec:prl}

A two-player turn-based \emph{arena} is a tuple $\ArenaName = \tuple
{\PosSet[\PlrSym]} {\PosSet[\OppSym]} {\MovRel}$, with $\PosSet[\PlrSym] \cap
\PosSet[\OppSym] = \emptyset$ and $\PosSet \defeq \PosSet[\PlrSym] \cup
\PosSet[\OppSym]$, such that $\tuple {\PosSet}{\MovRel}$ is a finite directed
graph without sinks.
$\PosSet[\PlrSym]$ (\resp, $\PosSet[\OppSym]$) is the set of positions of the
Player (\resp, the Opponent) and $\MovRel \subseteq \PosSet \times \PosSet$ is
a left-total relation describing all possible moves.
A \emph{path} in $\VSet \subseteq \PosSet$ is a finite or infinite sequence
$\pthElm \in \PthSet(\VSet)$ of positions in $\VSet$ compatible with the move
relation, \ie, $(\pthElm_{i}, \pthElm_{i + 1}) \in \MovRel$, for all $i \in
\numco{0}{\card{\pthElm} - 1}$.
% For a finite path $\pthElm$, with $\lst{\pthElm}$ we denote the last position
% of $\pthElm$.
%
A positional \emph{strategy} for player $\alpha \in \{ \PlrSym, \OppSym \}$ on
$\VSet \subseteq \PosSet$ is a function $\strElm[\alpha] \in
\StrSet[\alpha](\VSet) \subseteq (\VSet \cap \PosSet[\alpha]) \to \VSet$,
mapping each $\alpha$-position $\posElm$ in $\VSet$ to position
$\strElm[\alpha](\posElm)$ compatible with the move relation, \ie, $(\posElm,
\strElm[\alpha](\posElm)) \in \MovRel$.
With $\StrSet[\alpha](\VSet)$ we denote the set of all $\alpha$-strategies on
$\VSet$.
When talking about players, with $\dual{\alpha}$ we will refer to the opponent player of $\alpha$.
A \emph{play} in $\VSet \subseteq \PosSet$ from a position $\posElm \in \VSet$
\wrt a pair of strategies $(\strElm[\PlrSym], \strElm[\OppSym]) \in
\StrSet[\PlrSym](\VSet) \times \StrSet[\OppSym](\VSet)$, called
\emph{$((\strElm[\PlrSym], \strElm[\OppSym]), \posElm)$-play}, is a path
$\pthElm \in \PthSet(\VSet)$ such that $(\pthElm)_{0} = \posElm$ and, for all $i
\in \numco{0}{\card{\pthElm} - 1}$, if $(\pthElm)_{i} \in \PosSet[\PlrSym]$ then
$(\pthElm)_{i + 1} = \strElm[\PlrSym]((\pthElm)_{i})$ else $(\pthElm)_{i + 1} =
\strElm[\OppSym]((\pthElm)_{i})$.
The \emph{play function} $\playFun : (\StrSet[\PlrSym](\VSet) \times
\StrSet[\OppSym](\VSet)) \times \VSet \to \PthSet(\VSet)$ returns, for each
position $\posElm \in \VSet$ and pair of strategies $(\strElm[\PlrSym],
\strElm[\OppSym]) \in \StrSet[\PlrSym](\VSet) \times \StrSet[\OppSym](\VSet)$,
the maximal $((\strElm[\PlrSym], \strElm[\OppSym]), \posElm)$-play
$\playFun((\strElm[\PlrSym], \strElm[\OppSym]), \posElm)$.
Given a partial function $\fFun : \ASet \pto \BSet$, $\dom{\fFun}\subseteq \ASet$ and $\rng{\fFun} \subseteq \BSet$ we indicate the domain and the range of $\fFun$.

A \emph{parity game} is a tuple $\GameName = \tuple {\ArenaName} {\PrtSet}
{\prtFun} \in \PG$, where $\ArenaName$ is an arena, $\PrtSet \subset \SetN$ is a
finite set of priorities, and $\prtFun : \PosSet \to \PrtSet$ is a
\emph{priority function} assigning a priority to each position.
The priority function can be naturally extended to games and paths as follows:
$\prtFun(\GameName) \defeq \max[\posElm \in \PosSet] \, \prtFun(\posElm)$; for a
path $\pthElm \in \PthSet$, we set $\prtFun(\pthElm) \defeq \max_{i \in
\numco{0}{\card{\pthElm}}} \, \prtFun((\pthElm)_{i})$, if $\pthElm$ is finite,
and $\prtFun(\pthElm) \defeq \limsup_{i \in \SetN} \prtFun((\pthElm)_{i})$,
otherwise.
A set of positions $\VSet \subseteq \PosSet$ is an $\alpha$-\emph{dominion},
with $\alpha \in \{ 0, 1 \}$, if there exists an $\alpha$-strategy
$\strElm[\alpha] \in \StrSet[\alpha](\VSet)$ such that, for all
$\dual{\alpha}$-strategies $\strElm[\dual{\alpha}] \in
\StrSet[\dual{\alpha}](\VSet)$ and positions $\posElm \in \VSet$, the induced
play $\pthElm = \playFun((\strElm[\PlrSym], \strElm[\OppSym]), \posElm)$ is
infinite and $\prtFun(\pthElm) \equiv_{2} \alpha$.
In other words, $\strElm[\alpha]$ only induces on $\VSet$ infinite plays whose
maximal priority visited infinitely often has parity $\alpha$.
The \emph{winning region} for player $\alpha \in \{ 0, 1 \}$ in game
$\GameName$, denoted by $\WinSet[\GameName][\alpha]$, is the greatest set of
positions that is also a $\alpha$-dominion in $\GameName$.
Since parity games are memoryless determined~\cite{EJ91}, meaning that from each
position one of the two players wins, the two winning regions of a game
$\GameName$ form a partition of its positions, \ie, $\WinSet[\GameName][\PlrSym]
\cup \WinSet[\GameName][\OppSym] = \PosSet[\GameName]$.

By $\GameName \!\setminus\! \VSet$ we denote the maximal subgame of $\GameName$
with set of positions $\PosSet'$ contained in $\PosSet \!\setminus\! \VSet$ and
move relation $\MovRel'$ equal to the restriction of $\MovRel$ to $\PosSet'$.
The $\alpha$-predecessor of $\VSet$, in symbols $\preFun[][\alpha](\VSet) \defeq
\set{ \posElm \in \PosSet[\alpha] }{ \MovRel(\posElm) \cap \VSet \neq \emptyset
} \cup \set{ \posElm \in \PosSet[\dual{\alpha}] }{ \MovRel(\posElm) \subseteq
\VSet }$, collects the positions from which player $\alpha$ can force the game
to reach some position in $\VSet$ with a single move.
The $\alpha$-attractor $\atrFun[][\alpha](\VSet)$ generalises the notion of
$\alpha$-predecessor $\preFun[][\alpha](\VSet)$ to an arbitrary number of moves.
Thus, it corresponds to the least fix-point of that operator.
When $\VSet = \atrFun[][\alpha](\VSet)$, player $\alpha$ cannot force any
position outside $\VSet$ to enter this set that is,therefore, called
$\alpha$-maximal.
For such a $\VSet$, the set of positions of the subgame $\GameName \setminus
\VSet$ is precisely $\PosSet \setminus \VSet$.
When the computation of the attractor is restricted to a given set of
positions $\XSet$, we will use the notation $\atrFun[][\alpha](\VSet,\XSet)$
which corresponds to the least fix-point of $\preFun[][\alpha](\VSet) \cap
\XSet$.
Finally, the set $\escFun[][\alpha](\VSet) \defeq \preFun[][\alpha](\PosSet
\setminus \VSet) \cap \VSet$, called the \emph{$\alpha$-escape} of $\VSet$,
contains the positions in $\VSet$ from which $\alpha$ can leave $\VSet$ in one
move.
Observe that all the operators and sets described above actually depend on the
specific game $\GameName$ they are applied to.
In the rest of the paper, we shall only add $\GameName$ as subscript of an
operator, \eg\ $\atrFun[\GameName][\alpha][\VSet]$, when the game is not clear
from the context.

% End of file Preliminaries.tex

%% file: PrioProm.tex
% Begin of file PrioProm.tex

\section{A Hybrid Priority-Promotion Algorithm}
\label{sec:hybalg}

We introduce the hybrid algorithm in three steps. In the first step (Section
\ref{sec:hybalg;sub:prtprm}), we introduce a variation of classic Priority
Promotion, which serves as the backbone of our hybrid algorithm in Section
\ref{sec:hybalg;sub:hybalg}.
We provide a recap of how Priority Promotion operates and an introduction to its
data structure, which we later extend for our hybrid algorithm.
In a nutshell, Priority Promotion accelerates the classic recursive algorithm,
by allowing to merge dominions in subgames that span non-adjacent recursive
calls, which is the essence of the promotion operations.
In the following subsection (Section \ref{sub:parys}), we outline Parys'
algorithm, which does not seek to identify \emph{all dominions} on a level, but merely
\emph{all small dominions} up to given bounds $b_0$ and $b_1$ for the dominions of Player $0$ and
Player $1$, respectively.
It truncates the size of the call tree by making all but one call with half the
precision for one of the players.
Here, we formulate the algorithm with a terminology analogous to Priority
Promotion, and present it in a form similar to our hybrid algorithm.

The two concepts of Priority Promotion and truncated tree size through limited
guarantees appear to be unlikely allies: not only does the presence of Parys'
sets with limited guarantees impede the promotion of dominions, any attempt to
promote sets with bounded guarantees are set to fail, when the bounds are larger
(and thus the required guarantees stronger) along the call tree.
In Section \ref{sec:hybalg;sub:hybalg}, we see that, when synthesising the
algorithms carefully, sets with the `region guarantees' from Priority Promotion
and with `bounded guarantees' from Parys' approach can co-exist, so long as they
are kept carefully apart and treated differently.

The resulting algorithm can identify dominions in many places, and these
dominions can be promoted.
This promotion can be to a set with `region guarantees' at a higher level, but
it can also be that the correct target is a set with `bounded guarantees' (which
works across levels because dominions have unbounded guarantees).
The identification of the right set to promote to, instead, remains fairly
similar to the way it is identified in classic Priority Promotion.
While sets with bounded guarantees cannot be promoted along the data structure
(which follows the call tree), they lose parts of their locality: positions can
be promoted into them, and, crucially, they do not prevent promotions to higher
levels.
This way, we can keep the Priority Promotion part, which usually carries the
main burden of solving the parity game and can play out its practical efficiency
in full, while we also retain the quasi-polynomial complexity from Parys'
algorithm~\cite{Par19}, bypassing the known hard cases for recursive algorithms.
For practical considerations, it is still computationally attractive to grow the
bounded sets more slowly: we found that some of the points where Parys'
algorithm applies a closure of sets with bounded guarantees are merely for the
convenience of the proof.
For efficiency, we have restricted the closure under attractor of these sets to
the places where it is necessary for correctness.

\input{PrioPromA}

\input{PrioPromB}

\input{PrioPromC}

% End of file PrioProm.tex

%% file: PrioPromA.tex
% Begin of file PrioPromA.tex

\subsection{The Priority-Promotion Approach}
\label{sec:hybalg;sub:prtprm}

The \emph{priority-promotion approaches}~\cite{BDM18,BDM18a} attack the problem
of solving a parity game $\GameName \in \PG$ by iteratively computing, one at a
time, a sequence of $\alpha$-dominions $\DSet[0][\alpha], \DSet[1][\alpha],
\ldots \subseteq \PosSet$, for some player $\alpha \in \SetB \defeq \{ \PlrSym,
\OppSym \}$.
These, indeed, are portions of the two winning regions, $\WinSet[\PlrSym]$ and
$\WinSet[\OppSym]$, that need to be identified.
The idea here is to start from a weaker notion, called \emph{quasi dominion},
and then compose them until a dominion is obtained.
The name of the approach comes precisely from the fact that this composition is
computed by applying the following operation of \emph{promotion}: given two
quasi dominions $\QSet[1]$ and $\QSet[2]$ to which some priorities $\prtElm[1] <
\prtElm[2]$ of the same parity are assigned, $\QSet[1]$ is combined with
$\QSet[2]$ by promoting the former to the priority of the latter.

Similarly to a dominion, a quasi dominion is a set of positions over which one
of the two players, called the leading player, has a strategy defined on that
set, whose induced plays, if infinite, are winning for that player.
As opposed to dominions, however, some of these plays may be finite, since the
opponent may have the possibility to escape from those positions towards a
different part of the game, hoping for a better outcome.

\begin{definition}[Quasi Dominion]
  \label{def:qsidom}
  A set of positions $\QSet \subseteq \PosSet$ is a \emph{quasi
  $\alpha$-dominion}, for some player $\alpha \in \{ \PlrSym, \OppSym \}$, if
  there exists an $\alpha$-strategy $\strElm[\alpha] \in
  \StrSet[][\alpha](\QSet)$, called \emph{$\alpha$-witness for $\QSet$}, such
  that, for all $\dual{\alpha}$-strategies $\strElm[\dual{\alpha}] \in
  \StrSet[][\dual{\alpha}](\QSet)$ and positions $\posElm \in \QSet$, the
  induced play $\pthElm = \playFun((\strElm[\PlrSym], \strElm[\OppSym]),
  \posElm)$, satisfies
  $\prtFun(\pthElm) \equiv_{2} \alpha$, if infinite.
\end{definition}

The usefulness of the above concept, in addition to the property of being suitably
composable, resides in the fact that quasi dominions are closed under inclusion.
Thus, when a closed subset of a quasi dominion is found, a dominion is identified.
Notice that, differently from the definition given in~\cite{BDM16}, no constraint is
given on the finite plays, since the type of quasi dominions of~\cite{BDM16} are not
closed under inclusion.
However, such a constraint is implicitly stated in Definition~\ref{def:regfun}
below.

\begin{theorem}[Induced Dominion]
  \label{thm:inddom}
  Let $\alpha \in \{ \PlrSym, \OppSym \}$ be a player, $\QSet \subseteq \PosSet$
  a quasi $\alpha$-dominion, $\strElm \in \StrSet[\alpha](\QSet)$ one of its
  $\alpha$-witnesses, and $\DSet \subseteq \QSet$ a subset such that
  $\strElm(\posElm) \in \DSet$, if $\posElm \!\in\! \PosSet[\alpha]$, and
  $\MovRel(\posElm) \!\subseteq\! \DSet$, otherwise, for all positions $\posElm
  \!\in\! \DSet$.
  Then, $\DSet$ is an $\alpha$-dominion.
\end{theorem}

Intuitively, the solution algorithms following this approach carry on the search
for a dominion by exploring a \emph{finite strict partial order} $\tuple
{\SttSet} {\isttElm} {\prec}$, whose elements, called \emph{states}, record
information about the quasi dominions computed up to a certain point.
In the \emph{initial state} $\isttElm$, the quasi dominions are initialised to
the sets of positions with the same priority.
At each step, a new quasi $\alpha$-dominion $\QSet$, for some player $\alpha \in
\SetB$, is \emph{extracted} from the current state $\sttElm$ and used to compute
a \emph{successor state} \wrt the order $\prec$, if $\QSet$ is \emph{open}, \ie,
it if is not an $\alpha$-dominion.
If, on the other hand, it is \emph{closed}, the search is over and $\QSet$ is
added to the portion of the winning region $\WinSet[\alpha]$ computed so far.

We start by describing a new priority-promotion algorithm that instantiates the
above partial order and serves as a basis for the hybrid approach presented
later in this section.
To do so, we first need to introduce few technical notions, all of which refer
to some fixed parity game $\GameName \in \PG$.
By $\PrtSet[\botSym] \defeq \PrtSet \cup \{ \botSym \}$ and $\PrtSet[\topSym]
\defeq \PrtSet \cup \{ \topZSym, \topOSym \}$ we denote the set of priorities in
$\GameName$ extended with the bottom symbol $\botSym$ and two top symbols
$\topZSym$ and $\topOSym$, one for each player.
The standard ordering $<$ on $\PrtSet$ is extended to these additional elements
in the natural way: $\botSym$ is the smallest element, while both $\topZSym$ and
$\topOSym$ are strictly greater than every other priority; we do not assume any
specific order between the two maximal elements, though, we consider $\topZSym$
even and $\topOSym$ odd.

The first step in the formalisation of the notion of state requires the concept
of \emph{promotion function}, which represents the backbone of the algorithm,
being the data structure to which the promotion operation is applied.
Intuitively, it is a partial function from positions to priorities that
over-approximates the priority function of the game.

\begin{definition}[Promotion Function]
  \label{def:prmfun}
  A \emph{promotion function} $\prmFun \in \PrmSet \defeq \PosSet \pto
  \PrtSet[\topSym]$ is a partial function such that $\prmFun(\posElm) \geq
  \prtFun(\posElm)$, for every position $\posElm \in \dom{\prmFun}$.
\end{definition}

In the following, we adopt the same notation as in~\cite{BDM18}.
Given a promotion function $\prmFun \in \PrmSet$ and a priority $\prtElm \in
\PrtSet$, we denote with $\prmFun[][(\sim \prtElm)]$, for $\sim\: \in \{ <,
\leq, \equiv_{2}, \geq, > \}$, the function obtained by restricting the domain
of $\prmFun$ to those positions $\posElm \in \dom{\prmFun}$ whose priority
$\prmFun(\posElm)$ satisfies the relation $\prmFun(\posElm) \sim \prtElm$, \ie,
$\prmFun[][(\sim \prtElm)] \defeq \prmFun \rst \set{ \posElm \in \dom{\prmFun}
}{ \prmFun(\posElm) \sim \prtElm }$, where $\rst$ is the standard operation of
domain restriction.
We may also use Boolean combinations of the above restrictions, as in
$\prmFun[][(\equiv_{2} \alpha) \wedge (\geq \prtElm)]$.
By $\HSet[\prmFun][\alpha] \defeq \dom{\prmFun[][(\equiv_{2} \alpha)]}$ we
denote the set of positions in $\prmFun$ with a priority congruent to $\alpha
\in \SetB$ and with $\HSet[\prmFun][\alpha, \prtElm] \defeq
\dom{\prmFun[][(\equiv_{2} \alpha) \wedge (\geq \prtElm)]}$ its subset with
priorities greater than or equal to $\prtElm$.

A state encodes information about the quasi dominions computed up to a certain
point of the computation.
To this end, we require all positions in a promotion function $\prmFun$ with
priority of parity $\alpha \in \SetB$, \ie, the set $\HSet[\prmFun][\alpha]$, to
form a quasi $\alpha$-dominion.
Moreover, the idea is to store all $\alpha$-dominions already identified by
associating them with the corresponding maximal priority $\topSym[\alpha]$.

\begin{definition}[Quasi-Dominion Function]
  \label{def:qsidomfun}
  A \emph{quasi-dominion function} $\qdomFun \in \QDomSet \subseteq \PrmSet$ is
  a promotion function satisfying the following conditions, for every $\alpha
  \in \SetB$:
  \begin{inparaenum}[1)]
    \item\label{def:qsidomfun(qsi)}
      the set $\HSet[\qdomFun][\alpha]$ is a quasi $\alpha$-dominion;
    \item\label{def:qsidomfun(dom)}
      the set $\qdomFun[][-1](\topSym[\alpha])$ is an $\alpha$-dominion.
  \end{inparaenum}
\end{definition}

An important property of dominions is that the extension of an $\alpha$-dominion $\QSet$
by means of its $\alpha$-attractor is still an $\alpha$-dominion.
This property, however is not enjoyed by arbitrary quasi dominions.
Indeed, there may even be cases where the $\alpha$-attractor of a quasi
$\alpha$-dominion is a $\dual{\alpha}$-dominion.
Moreover, to efficiently verify whether a quasi dominion is actually a dominion,
an explicit representation of one of its witnesses is usually required.
To overcome these complications, we consider a subclass of quasi dominions that
meets the following requirements:
\begin{inparaenum}[1)]
  \item
    the set $\escFun(\QSet, \strElm) \defeq \set{ \posElm \in \PosSet[\alpha]
    \cap \QSet }{ \strElm(\posElm) \not\in \QSet }$ of $\alpha$-positions, which
    leave a quasi $\alpha$-dominion $\QSet \subseteq \PosSet$ by following one
    of its $\alpha$-witnesses $\strElm \in \StrSet[][\alpha](\QSet)$, is a
    subset of the $\dual{\alpha}$-escape positions
    $\escFun[][\dual{\alpha}](\QSet)$;
  \item
    all these $\dual{\alpha}$-escape positions have priorities congruent to
    $\alpha$ and greater than the ones of the positions that can be attracted by $\alpha$ to $\QSet$.
\end{inparaenum}
The first requirement ensures that, to verify whether $\QSet$ is an
$\alpha$-dominion, it suffices to check for the emptiness of
$\escFun[][\dual{\alpha}](\QSet)$.
The second one, instead, can be exploited to regain closure under extension by
$\alpha$-attractor.

\begin{definition}[Region Function]
  \label{def:regfun}
  A \emph{region function} $\regFun \in \RegSet \subseteq \QDomSet$ is a
  quasi-dominion function satisfying the following conditions, for every $\alpha
  \in \SetB$:
  \begin{inparaenum}[1)]
    \item\label{def:qsidomfun(wit)}
      there exists an $\alpha$-witness $\strElm[\alpha] \in
      \StrSet[][\alpha](\HSet[\regFun][\alpha])$ for $\HSet[\regFun][\alpha]$
      such that $\escFun(\HSet[\regFun][\alpha, \prtElm], \strElm[\alpha])
      \subseteq \escFun[][\dual{\alpha}](\HSet[\regFun][\alpha])$, for all
      $\prtElm \in \rng{\regFun}$, with $\prtElm \equiv_{2} \alpha$;
    \item\label{def:qsidomfun(prt)}
      $\prtElm \leq \prtFun(\posElm) \equiv_{2} \alpha$, for all $\prtElm \in
      \rng{\regFun}$, with $\prtElm \equiv_{2} \alpha$, and $\posElm \in
      \escFun[][\dual{\alpha}](\HSet[\regFun][\alpha, \prtElm])$.
  \end{inparaenum}
\end{definition}

Notice that, every set $\HSet[\regFun][\alpha, \prtElm]$, with $\prtElm \in
\rng{\regFun}$, is a quasi $\alpha$-dominion, being a subset of the quasi
$\alpha$-dominion $\HSet[\regFun][\alpha]$.
Also, it is immediate to see that the priority function $\prtFun$ of a given
parity game $\GameName$ is always a region function.
Indeed, it is trivially a promotion function.
Moreover, the positions with a priority of parity $\alpha$, \ie,
$\HSet[\prtFun][\alpha]$, form a quasi $\alpha$-dominion with $\alpha$-witness
any strategy that always chooses to remain inside the set, if allowed by the
move relation.
Thus, it is a quasi-dominion function as well.
Finally, since $\HSet[\prtFun][\alpha, \prtElm]$ cannot contain positions of
parity $\dual{\alpha}$ and thanks to the way the $\alpha$-witness is chosen, it
is clear that $\prtFun$ also satisfies the conditions of
Definition~\ref{def:regfun}.

At this point, we have the technical tools to introduce the \emph{search space}
that instantiates the finite strict partial order described in the intuitive
explanation of the approach.
In particular, to account for the current status of the search of a dominion in
a game $\GameName$, we define a \emph{state} $\sttElm$ as a pair, comprising a
region function $\regFun$ and a priority $\prtElm$, with the idea that
\begin{inparaenum}[1)]
  \item
    all quasi $\alpha$-dominions computed so far are contained in
    $\HSet[\regFun][\alpha, \qprtElm]$, for some $\qprtElm > \prtElm$,
  \item
    the current quasi dominion to focus on is contained in $\regFun$ at priority
    $\prtElm$, and
  \item
    all positions with priorities smaller than or equal to $\prtElm$ correspond
    to the portion of the game that has still to be processed.
\end{inparaenum}
The initial state is composed of the priority function $\prtFun$ of the game and
its maximal priority $\prtFun(\GameName)$.
Finally, we assume that a state $\sttElm[1]$ is lower than another state
$\sttElm[2]$ \wrt the partial order relation $\prec$, if the set of unprocessed
positions in $\sttElm[1]$ is a subset of those in $\sttElm[2]$.

\begin{definition}[Search Space]
  \label{def:srcspc}
  A \emph{search space} is a tuple $\SName \defeq \tuple {\SttSet} {\isttElm}
  {\prec}$, whose three components are defined as follows:
  \vspace{-0.5em}
  \begin{enumerate}
  %\begin{inparaenum}[1)]
    \item
      $\SSet \subseteq \RegSet \times \PrtSet[\bot]$ is the set of all pairs
      $\sttElm \defeq (\regFun, \prtElm)$, called \emph{states}, where
      $\dom{\regFun} = \PosSet$;
      for every state $\sttElm \in \SttSet$, we set
      \begin{inparaenum}[(i)]
        \item
          $\alpha_{\sttElm} \defeq \prtElm \bmod{2}$,
        \item
          $\HSet[\sttElm][\alpha] \defeq \HSet[\regFun][\alpha]$,
        \item
          $\HSet[\sttElm][\alpha, \qprtElm] \defeq \HSet[\regFun][\alpha,
          \qprtElm]$,
        \item
          $\RSet[\sttElm] \defeq \regFun[][-1](\prtElm)$, and
        \item
          $\LSet[\sttElm] \defeq \dom{\regFun[][(\leq \prtElm)]}$;
      \end{inparaenum}
    \item
      $\isttElm \defeq (\prtFun, \prtFun(\GameName))$ is the \emph{initial
      state};
    \item
      $\sttElm[1] \prec \sttElm[2]$ if either $\LSet[{\sttElm[1]}] \subset
      \LSet[{\sttElm[2]}]$ or $\LSet[{\sttElm[1]}] = \LSet[{\sttElm[2]}]$ and
      $\prtElm[{\sttElm[1]}] < \prtElm[{\sttElm[2]}]$.
  %\end{inparaenum}
  \end{enumerate}
\end{definition}

Given a state $\sttElm \in \SttSet$, we refer to $\LSet[\sttElm]$ as the
\emph{local area}, \ie, the set of unprocessed positions yet to be analysed.
This also includes the quasi $\alpha_{\sttElm}$-dominion $\RSet[\sttElm]$,
called \emph{region}, on which the next step of the search will focus.
The two quasi dominions $\HSet[\sttElm][\PlrSym]$ and $\HSet[\sttElm][\OppSym]$
partition the entire set of positions in the game, while
$\HSet[\sttElm][\PlrSym, \qprtElm]$ and $\HSet[\sttElm][\OppSym, \qprtElm]$
represent the portions of these quasi dominions to which the region function
has assigned a priority at least equal to $\qprtElm \in \PrtSet$.
Notice that the pseudo-priority $\bot$ is used to indicate the situation where
all positions have been processed, which corresponds to an empty local area.
\begin{wrapfigure}[14]{R}{0.532\textwidth}
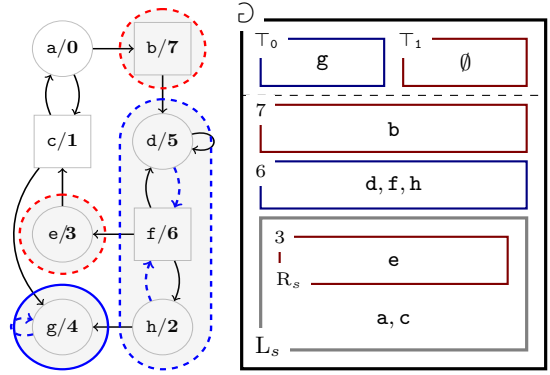

  %\vspace{-1.30em}
  \scalebox{0.75}[0.7]{\figexma}
  \hspace{-1em}
  \scalebox{1}[1]{\figsttspcexp}
  %\vspace{-1.5em}
  \caption{\label{fig:sttspcexp} \small A game and a corresponding state
   representation.}
  \vspace{-0em}
\end{wrapfigure}

\indent

To exemplify the above notions, consider the game depicted in
Figure~\ref{fig:sttspcexp}, where circled shaped positions belong to
Player~$\PlrSym$ and square shaped ones to Player~$\OppSym$.
Clearly, $\gSym$ and $\hSym$ are won by Player~$\PlrSym$, while the rest of the
game is won by the opponent.
At the state $\sttElm = (\regFun, 3)$, where $\regFun = \{ \aSym \!\mapsto\! 0;
\cSym \!\mapsto\! 1; \eSym \!\mapsto\! 3; \dSym, \fSym, \hSym \!\mapsto\! 6;
\allowbreak \bSym \!\mapsto\! 7; \gSym \!\mapsto\! \topZSym \}$, the local area
$\LSet[\sttElm]$ contains the positions $\aSym$, $\cSym$, and $\eSym$.
Of these only $\eSym$ is part of the current region $\RSet[\sttElm]$.
The quasi $\PlrSym$-dominion $\HSet[\sttElm][\PlrSym]$ contains the positions
$\aSym$, $\dSym$, $\fSym$, $\hSym$, and $\gSym$, while the quasi
$\OppSym$-dominion $\HSet[\sttElm][\OppSym]$ takes the remaining ones, namely
$\bSym$, $\cSym$, and $\eSym$.
Position $\gSym$ forms a $\PlrSym$-dominion on its own, represented in the
picture by the solid closed line.

Apart from this position, all the other ones are contained in open quasi
dominions, indicated, instead, by the dashed closed lines.
For example, the set $\regFun[][-1](6) = \{ \dSym, \fSym, \hSym \}$ is a quasi
$\PlrSym$-dominion, since, if Player $\OppSym$ decides to remain inside, the
adversary wins the play.
However, Player $\OppSym$ also has the choice to escape from position $\fSym$
moving to $\eSym$, \ie, $\escFun[][1](\regFun[][-1](6)) = \{ \fSym \}$.
Similarly, $\escFun[][0](\regFun[][-1](7)) = \{ \bSym \}$.
Finally, notice that $\HSet[\sttElm][\PlrSym, 4] = \{ \dSym, \fSym, \gSym, \hSym
\}$ and $\HSet[\sttElm][\OppSym, 4] = \{ \bSym \}$.
During the exploration of the search space, a priority-promotion algorithm
typically traverses several types of states, some of which enjoy important
properties that need to be explicitly identified, as they are exploited during
the search for a dominion.
Given a player $\alpha \in \SetB$, we say that a state $\sttElm \in \SttSet$ is
\emph{$\alpha$-maximal}, if the quasi $\alpha$-dominion $\HSet[\sttElm][\alpha]
\setminus \LSet[\sttElm]$ is $\alpha$-maximal \wrt $\LSet[\sttElm]$, \ie, the
$\alpha$-attractor $\atrFun[][\alpha](\HSet[\sttElm][\alpha] \setminus
\LSet[\sttElm], \LSet[\sttElm])$ to $\HSet[\sttElm][\alpha] \setminus
\LSet[\sttElm]$ of positions from the local area $\LSet[\sttElm]$ is empty.
If $\sttElm$ is $\alpha$-maximal \wrt both players $\alpha \in \SetB$, we simply
say that it is \emph{maximal} and denote by $\SttSet[\MSym] \subseteq \SttSet$
the corresponding subset of states and by $\GameName[\sttElm] \defeq \GameName
\setminus \dom{\regFun[\sttElm][(> {\prtElm[\sttElm]})]}$ the induced subgame.
A maximal state $\sttElm$ is \emph{strongly maximal}, if the current region
$\RSet[\sttElm]$ is $\alpha_{\sttElm}$-maximal \wrt $\LSet[\sttElm]$.
By $\SttSet[\SSym] \subseteq \SttSet[\MSym]$ we denote the set of strongly
maximal states.
Recall that region $\RSet[\sttElm]$ of a state $\sttElm$ is contained in the
quasi $\alpha_{\sttElm}$-dominion
$\HSet[\sttElm][\alpha_{\sttElm}, {\prtElm[\sttElm]}]$. We say that $\sttElm$ is
\emph{open} if the opponent $\dual{\alpha}[\sttElm]$ can escape from
$\HSet[\sttElm][ \alpha_{\sttElm}, {\prtElm[\sttElm]}]$ starting from
$\RSet[\sttElm]$ using a single move, \ie, if
$\RSet[\sttElm] \cap
\escFun[][ {\dual{\alpha}[\sttElm]} ](\HSet[\sttElm][\alpha_{\sttElm},
{\prtElm[\sttElm]}]) \neq \emptyset$. In this case, the opponent may escape from
$\RSet[\sttElm]$ by either moving to the remaining portion of local area
$L_{\sttElm}\setminus\RSet[\sttElm]$ or to the quasi
$\dual{\alpha}[\sttElm]$-dominion
$\HSet[\sttElm][ {\dual{\alpha}[\sttElm]}, {\prtElm[\sttElm]} ]$.  The state is
said to be \emph{closed}, otherwise.
For technical convenience, a state with an empty region is always considered
open.
Finally, a closed state $\sttElm$ is \emph{promotable}, if it
$\dual{\alpha}[\sttElm]$-maximal and $\RSet[\sttElm]$ is
$\alpha_{\sttElm}$-maximal \wrt $\LSet[\sttElm]$.
By $\SttSet[\PSym] \subseteq \SttSet$ we denote the set of promotable states.

%\begin{wrapfigure}[16]{R}{0.358\textwidth}
%  \vspace{-2.2em}
%  \algsolexpl
%  \vspace{-0em}
%\end{wrapfigure}
%\indent
The main function $\solFun$ of the new priority-promotion-based approach, called \emph{recursive priority
promotion} (\RPP, for short), is reported in Algorithm~\ref{alg:expsol}.
The auxiliary function $\NxtFun$ and the two procedures $\MaxFun$ and
$\PrmFun$ are, instead, reported in Appendix~\ref{app:c}.
The function $\solFun$ assumes the input state $\sttElm$ to be maximal, \ie,
$\sttElm \in \SttSet[\MSym]$.
At Line~1 it checks if there are still unprocessed positions in the game, namely
if the priority of the current state is different from $\bot$.
If this is the case, Line~2 maximises the region of the current state, namely
$\RFun[\sttElm] \defeq \regFun[][-1](\prtElm)$, by computing its
$\alpha_{\sttElm}$-attractor, so that the resulting set is
$\alpha_{\sttElm}$-maximal and, therefore, $\sttElm$ becomes strongly maximal,
\ie, $\sttElm \in \SttSet[\SSym]$.

\begin{wrapfigure}[16]{R}{0.358\textwidth}
  \vspace{-2.5em}
  \algsolexpl
  \vspace{-0em}
\end{wrapfigure}
\indent
For convenience, we abbreviate the update of some component in a state
$\sttElm$, say component $\RFun[\sttElm]$ for instance, simply as
$\RFun[\sttElm] \gets \mthelm{exp}$, for some expression $\mthelm{exp}$.
Therefore, the instruction at Line~2 updates the state $\sttElm$ by replacing
the original region $\RFun[\sttElm]$ with $\atrFun[ {\GameName[\sttElm]}
][\alpha_{\sttElm}][ {\RFun[\sttElm]} ]$ within the state.
If the resulting state $\sttElm$ is closed, it is also promotable, \ie,
$\sttElm \in \SttSet[\PSym]$, being maximal by hypothesis, and, therefore, a
promotion can be applied at Line~8, by means of a call to procedure $\PrmFun$.
If, instead, $\sttElm$ is open, which means that the opponent can escape from
$\RFun[\sttElm]$ moving outside the quasi $\alpha_{\sttElm}$-dominion
$\HSet[\sttElm][\alpha_{\sttElm}, {\prtElm[\sttElm]}]$, the algorithms proceeds
to analyse the part of the game still unprocessed.
To do this, we first compute the next state by means of $\NxtFun(\sttElm)$,
which simply identifies the next priority to consider, namely the maximal
priority of the unprocessed positions.
The resulting state is then given as input to the recursive call at Line~4.
Once the recursive call completes, the state is updated with the new region
function returned by the call.
The new state $\sttElm$ is such that the local area $\LSet[\sttElm]$ coincides
with $\RSet[\sttElm]$, since the recursive call ends after analysing all the
previously unprocessed positions.
As a consequence, either the opponent cannot escape from $\RFun[\sttElm]$
anymore or it can only move to its own quasi dominion $\HSet[\sttElm][
{\dual{\alpha}[\sttElm]}, {\prtElm[\sttElm]} ]$.
Line~5 checks which one of the two possibilities occurs.
In the first case, the new state $\sttElm$ is closed, hence
$\dual{\alpha}[\sttElm]$-maximal.
Moreover, since $\LSet[\sttElm] = \RSet[\sttElm]$, the region $\RFun[\sttElm]$
cannot attract any other positions and is, therefore,
$\alpha_{\sttElm}$-maximal.
This means that $\sttElm$ is promotable, \ie, $\sttElm \in \SttSet[\PSym]$, and
Line~7 promotes the region within the quasi $\alpha_{\sttElm}$-dominion.
If, on the other hand, $\sttElm$ is still open, then the opponent can escape to
$\HSet[\sttElm][ {\dual{\alpha}[\sttElm]}, {\prtElm[\sttElm]} ]$ from some
positions in $\RFun[\sttElm]$.
This means that the current state is not $\dual{\alpha}[\sttElm]$-maximal and
Line~6 fixes this by calling the procedure $\MaxFun$.
The aim of this function is to reestablish maximality of the quasi dominions
$\HSet[\sttElm][ \alpha_{\sttElm}, {\prtElm[\sttElm]}]$ and $\HSet[\sttElm][
{\dual{\alpha}[\sttElm]}, {\prtElm[\sttElm]} ]$ associated with the state
$\sttElm$.
This is done by attracting positions from the current region $\RFun[\sttElm] =
\LFun[\sttElm]$.
The surviving positions in $\RFun[\sttElm]$, if any, need not form a quasi
$\alpha_{\sttElm}$-dominion anymore and are set by $\MaxFun$ to their original
priority according to the priority function $\prtFun$ of the game.
In any case, when the computation reaches Line~9, whether coming from Line~6,
Line~7 or Line~8, the state $\sttElm$ is maximal, \ie, $\sttElm \in
\SttSet[\MSym]$, and a second, and final, recursive call is performed on
$\sttElm$ to process the remaining positions in $\LFun[\sttElm]$, if any.
We refer to Appendix~\ref{app:c} for the detailed description of the auxiliary function $\NxtFun$ and the two procedures $\MaxFun$ and $\PrmFun$.

% End of file PrioPromA.tex

%% file: PrioPromB.tex
% Begin of file PrioPromB.tex

\subsection{Parys' Algorithm}
\label{sub:parys}
In order to obtain a quasi-polynomial time priority-promotion-based solution
procedure, we entangle the algorithm of the previous subsection with Parys'
idea~\cite{Par19} to suitably truncate the recursion tree.
Naturally, cutting some of the recursive calls may prevent us from deciding the
winner for some of the positions with certainty.
These intermediate results are thus \emph{undetermined} (we use a function
$\undFun$, read `undetermined', to refer to these results).
Parys' contribution was to design the truncation in a way that offers bounded
guarantees, namely that the undetermined sets contain all small dominions of one player and
do not intersect with small dominions of the other.
The \emph{bounds} up to which these limited guarantees hold are shed quickly in
the call tree: most of the calls are made with half precision (meaning that one
of the bounds is halved) and only one is made with full precision, meaning
that both bounds are kept.
A first step in the integration of Parys' approach with Priority Promotion is to
formulate it in the same terms and to %use the opportunity to
introduce the
notation needed when hybridising the approaches.
To this end, we assume $\USet[\sttElm] \defeq
\undFun[\sttElm][-1](\prtElm[\sttElm])$ is the set of undetermined positions at
a certain stage $\sttElm$ of the search.
We require that $\USet[\sttElm]$ satisfies the following property: it must
contain all the $\dual{\alpha}$-dominions of size no greater than a
%less than or equal to some
given bound $\bElm[\alpha]$ and it cannot intersect any $\alpha$-dominion of
size no greater than 
%less than or equal to
a second given bound $\bElm[\dual{\alpha}]$.

Just as pure Priority Promotion does not use undetermined positions, Parys'
approach does not use regions (beyond the attractor of the nodes with highest
priority).
Consequently, little happens to the region functions in our representation of
Parys' algorithm: positions that are added to $\USet[\sttElm]$ are removed from
the region function, and when $\USet[\sttElm]$ is destroyed, they are added
(with their native priorities) back to $\regFun[\sttElm]$.

\begin{wrapfigure}[21]{R}{0.436\textwidth}
  \vspace{-1.20em}
  \algsolparl
  %\vspace{-1.10em}
  \algsolparhalf
  \vspace{-0em}
\end{wrapfigure}

We only outline the principles here together with slight generalisations of the
standard lemmas employed in each step, required later for hybrid approach.
The main function $\solFun$ and the halved solver $\hsolFun$ are reported in Algorithm~\ref{alg:parsol} and Algorithm~\ref{alg:parsolhalf}. The procedure $\MaxFun$ and the auxiliary functions $\NxtFun$, $\HalfFun$, and $\UndFun$ are, instead, reported in Appendix~\ref{app:c}.
In $\hsolFun$, the attractor of the positions with maximal priority is removed,
and the recursive call of $\solFun$ with half precision adds all small dominions $\leq \floor{
\frac{\bSym[ {\dual{\alpha}[\sttElm]} ]}{2} }$ (but no part of any small region
($\leq \bSym[\alpha_{\sttElm}]$) of player $\alpha$) of the remaining subgame to
$\USet[\sttElm]$.
Lemma \ref{lem:no-p} provides that, if such an $\overline\alpha_s$ dominion was contained in the game before removing $\RSet[\sttElm]$ then a sub-dominion of it still remains after $\RSet[\sttElm]$ is removed.

\begin{restatable}{lemma}{noattraction}
  \label{lem:noattraction}
If a dominion $D$ for Player $\alpha$ in a
game $\GameName$ does not intersect
with a set $S$ of nodes, then it does not intersect with
$\atrFun_{\GameName}^{\overline{\alpha}}(S)$ either.
\end{restatable}

For $\alpha \equiv_2 p$, a $p$-Region is a quasi $\alpha$-dominion, such that
(1) all nodes have priority $\leq p$, and (2) all escape positions have priority
$p$.

\begin{restatable}{lemma}{nop}
  \label{lem:no-p}
If the highest priority in a non-empty dominion $D$ for player $\alpha$ is $p$
and it intersects with a $q$-region $R$ for $q \not\equiv_2 \alpha$ and $q \geq
p$, %then $\alpha$'s parity (i.e.\ $\alpha \neq p \mod 2$), 
then
$D$
contains a non-empty sub-dominion that does not intersect with
$\atrFun_{\GameName}^{\overline{\alpha}}(R)$.
\end{restatable}

These dominions are then closed under attractor by a call of $\MaxFun$, until
$\sttElm$ does not change, and thus until no $\dual{\alpha}[\sttElm]$
dominion $\leq \lfloor \frac{b_{\dual{\alpha}[\sttElm]}}{2}\rfloor$ is left.
This guarantee is used in $\solFun$ to ensure that a dominion $\DSet$ of player
$\dual{\alpha}[\sttElm]$ is in $\USet[\sttElm]$ at the end of the function:
all $\dual{\alpha}[\sttElm]$ dominions returned in the recursive call of line 5
are  $> \lfloor \frac{b_{\dual{\alpha}[\sttElm]}}{2}\rfloor$.
Closing them under attractor through a call of $\MaxFun$ (line 6) leaves, by
Lemma \ref{lem:noattraction}, a (possibly empty) $\dual{\alpha}[\sttElm]$
dominion of size $\leq \lfloor \frac{b_{\dual{\alpha}[\sttElm]}}{2}\rfloor$,
when $\hsolFun$ is called for the second time, and thus fully included in
$U_s$, by means of $\UndFun$, that simply updates the data structures $\regFun[\sttElm]$ and $\undFun[\sttElm]$.
After Parys' Solver is run (with full precision) on a game with maximal priority
$p_{\max}$ identified by the function $\NxtFun$, we have that $u^{-1}(p_{\max}+1)$ contains the winning region of player
$p_{\max} \mod 2$.

\subsection{A Hybrid Algorithm}
\label{sec:hybalg;sub:hybalg}

In our hybrid approach, we have to synthesise the use of regions and the use of
undetermined sets.
To formalise this intuition, a state for the hybrid algorithm needs to embed
quite a bit of additional information \wrt a simple state of \RPP.
It obviously contains both a region function $\regFun$, tracking the quasi
dominions already analysed, and the current priority $\prtElm$.
It also features an additional promotion function $\undFun$, used to maintain
the set of undetermined positions, which only satisfy the relative guarantees
mentioned above, the priority of the caller used in the update of this function,
and two numbers $\bElm[\PlrSym]$ and $\bElm[\OppSym]$, representing the two
bounds \wrt which the guarantees are expressed.
The initial state, from which the search starts, contains, as is the case in the
exponential algorithm, the priority function $\prtFun$ and the maximal priority
$\prtFun(\GameName)$.
In addition, the two bounds are both set to the number of positions in the game,
while the accessory function $\undFun$ is empty.
For technical convenience, we set the caller priority to $\topZSym$.
Finally, the ordering between the states is again defined in terms of the sets
of unprocessed positions in the two states.

\begin{definition}[Hybrid Search Space]
  \label{def:hybsrcspc}
  A \emph{hybrid search space} is a tuple $\SName \defeq \tuple {\SttSet}
  {\isttElm} {\prec}$, whose three components are defined as follows:
  \begin{enumerate}
    \item
      $\SttSet \subseteq (\RegSet \times \PrtSet[\bot]) \times (\PrmSet \times
      \PrtSet[\top] \times \SetN \times \SetN)$ is the set of all tuples
      $\sttElm \defeq ((\regFun, \prtElm), (\undFun, \cprtElm, \bElm[\PlrSym],
      \bElm[\OppSym]))$, called \emph{hybrid states}, where:
      \begin{enumerate}[a)]
        \item\label{def:hybsrcspc(par)}
          $\dom{\regFun} \!\cap\! \dom{\undFun} \!=\! \emptyset$, $\dom{\regFun}
          \!\cup\! \dom{\undFun} \!=\! \PosSet$, $\dom{(\regFun \!\cup\!
          \undFun)^{(> \prtElm) \land (< \cprtElm)}} \!=\! \emptyset$, $\prtElm
          \!<\! \cprtElm$;
        \item\label{def:hybsrcspc(und)}
          $\undFun[][-1](\topOSym) = \dom{\undFun[][(< \prtElm)]} = \emptyfun$
          and if $\cprtElm \neq \topZSym$ then $\undFun[][-1](\topZSym) =
          \emptyfun$;
        \item\label{def:hybsrcspc(esc)}
          $\dom{\undFun} \cap \escFun[][\dual{\alpha}](\HSet[\sttElm][\alpha,
          \qprtElm]) = \emptyset$, with $\HSet[\sttElm][\alpha, \qprtElm] \defeq
          \HSet[\regFun][\alpha, \qprtElm] \cup \HSet[\undFun][\dual{\alpha},
          \qprtElm]$, for all $\alpha \in \SetB$ and $\qprtElm \geq \prtElm$;
      \end{enumerate}
      for every state $\sttElm \in \SttSet$, we set
      \begin{inparaenum}[(i)]
        \item
          $\alpha_{\sttElm} \defeq \prtElm \bmod{2}$,
        \item
          $\HSet[\sttElm][\alpha] \defeq \HSet[\regFun][\alpha] \cup
          \HSet[\undFun][\dual{\alpha}]$,
        \item
          $\USet[\sttElm] \defeq \undFun[][-1](\prtElm)$,
        \item
          $\RSet[\sttElm] \defeq \regFun[][-1](\prtElm)$, and
        \item
          $\LSet[\sttElm] \defeq \dom{\regFun[][(\leq \prtElm)]}$;
      \end{inparaenum}
    \item
      $\isttElm \defeq ((\prtFun, \prtFun(\GameName)), (\emptyfun,
      %\topSym[\dual{\alpha}]
      \topSym[0], \card{\GameName}, \card{\GameName}))$ is the
      \emph{initial state};%, where $\alpha \defeq \prtFun(\GameName) \bmod{2}$;
    \item
      $\sttElm[1] \prec \sttElm[2]$ if either $\LSet[{\sttElm[1]}] \subset
      \LSet[{\sttElm[2]}]$ or $\LSet[{\sttElm[1]}] = \LSet[{\sttElm[2]}]$ and
      $\prtElm[{\sttElm[1]}] < \prtElm[{\sttElm[2]}]$.
  \end{enumerate}
\end{definition}

Intuitively, Item~\ref{def:hybsrcspc(par)} ensures that the set of positions in
the game is partitioned into two categories: \emph{(i)} those contained in the
region function $\regFun$, which are considered \emph{determined}, in the sense
that they belong to known quasi dominions; and \emph{(ii)} those contained in
the promotion function $\undFun$, which are \emph{undetermined}, since they form
sets that only satisfy the bounded guarantees.
Obviously the priority of the caller has to be higher than the current one and
no position can be associated with a priority between those two.
Moreover, since positions assigned to the two top pseudo-priorities must be
determined, as they form dominions, $\undFun$ cannot refer to those two values,
except for the outermost call, when $\cprtElm = \topZSym$.
In this case, indeed, any undetermined position is necessarily won by player
$\OppSym$, \ie, $\undFun[][-1](\topZSym) \subseteq \WinSet[\OppSym]$.
Moreover, all positions with priorities lower than the current one are still
unprocessed, therefore they cannot be undetermined.
Both these requirements are expressed by Item~\ref{def:hybsrcspc(und)}.
Finally, we need to ensure that a player cannot immediately escape from a set of
undetermined positions, as specified in Item~\ref{def:hybsrcspc(esc)}.
This property is crucial to maintain, after the update of the promotion function
$\undFun$, the implicit invariant stating that, if a strongly-maximal state is
closed, then it is promotable.

\begin{wrapfigure}[13]{R}{0.285\textwidth}
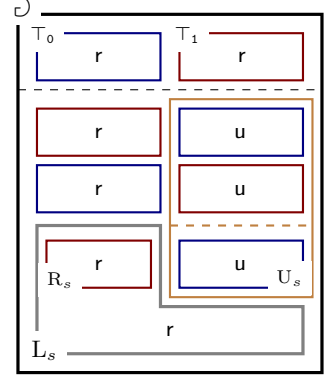

  \vspace{-3.00em}
  \scalebox{1}[1]{\figsttspchyb}
  \vspace{-0.5em}
  \caption{\label{fig:sttspchyb} \!\!\small The structure of a hybrid
    state.\!\!}
  \vspace{-0em}
\end{wrapfigure}

Figure~\ref{fig:sttspchyb} reports a graphical representation of the structure
of a hybrid state.
Most of the concepts and notation introduced for the states of \RPP have a
similar meaning and play a similar role for hybrid states.
In particular, given a hybrid state $\sttElm \in \SttSet$, the set
$\LSet[\sttElm]$ identifies the \emph{local area}, \ie, the set of positions yet
to analyse, while $\RSet[\sttElm]$ is the quasi $\alpha_{\sttElm}$-dominion,
called \emph{region}, included in $\LSet[\sttElm]$, which the algorithm is
currently focusing on.
Moreover, the two sets $\HSet[\sttElm][\PlrSym]$ and $\HSet[\sttElm][\OppSym]$
partition the game, while $\HSet[\sttElm][\PlrSym, \qprtElm]$ and
$\HSet[\sttElm][\OppSym, \qprtElm]$ represent the portions of these sets having
a priority, assigned either in $\regFun$ or in $\undFun$, at least equal to
$\qprtElm \in \PrtSet$.
As opposed to the previous notion of state, however, $\HSet[\sttElm][\PlrSym]$ and $\HSet[\sttElm][\OppSym]$ are not
necessarily quasi dominions, since they may include undetermined positions,
namely those contained in $\HSet[\uElm][\PlrSym, \qprtElm]$ and
$\HSet[\uElm][\OppSym, \qprtElm]$.
Only the two subsets $\HSet[{\regFun[\sttElm]}][\PlrSym]$ and
$\HSet[{\regFun[\sttElm]}][\OppSym]$, as well as their relativised versions
$\HSet[{\regFun[\sttElm]}][\PlrSym, \qprtElm]$ and
$\HSet[{\regFun[\sttElm]}][\OppSym, \qprtElm]$, are known to be quasi dominions.

\noindent
Given a player $\alpha \in \SetB$, we say that a hybrid state $\sttElm \in
\SttSet$ is \emph{$\alpha$-maximal}, if the quasi $\alpha$-dominion
$\HSet[\sttElm][\alpha] \setminus \LSet[\sttElm]$ is $\alpha$-maximal \wrt
$\LSet[\sttElm]$.
If $\sttElm$ is $\alpha$-maximal \wrt both players $\alpha \in \SetB$, we say
that it is \emph{maximal}.
We denote with $\SttSet[\MSym] \subseteq \SttSet$ the set of maximal hybrid
states and with $\GameName[\sttElm] \defeq \GameName \setminus
\dom{\regFun[\sttElm][(> {\prtElm[\sttElm]})] \cup \undFun[\sttElm]}$ the
induced subgame over the local area $\LSet[\sttElm]$.
A maximal hybrid state $\sttElm$ is \emph{strongly maximal}, if the current
region $\RSet[\sttElm]$ is $\alpha_{\sttElm}$-maximal \wrt $\LSet[\sttElm]$ and
the quasi $\dual{\alpha}[\sttElm]$-dominion $\HSet[\sttElm][
{\dual{\alpha}[\sttElm]} ] \setminus (\LSet[\sttElm] \cup \USet[\sttElm])$ is
$\dual{\alpha}[\sttElm]$-maximal \wrt $\LSet[\sttElm] \cup \USet[\sttElm]$.
By $\SttSet[\SSym] \subseteq \SttSet[\MSym]$ we denote the set of strongly
maximal hybrid states.
Again, we say that $\sttElm$ is \emph{open} if $\RSet[\sttElm] \cap \escFun[][
{\dual{\alpha}[\sttElm]} ](\HSet[\sttElm][\alpha_{\sttElm}, {\prtElm[\sttElm]}])
\neq \emptyset$, and we say that it is \emph{closed}, otherwise.
For technical convenience, we always consider a hybrid state with an empty
region open.
Finally, a closed hybrid state $\sttElm$ is \emph{promotable}, if it is
$\dual{\alpha}[\sttElm]$-maximal and $\RSet[\sttElm]$ is
$\alpha_{\sttElm}$-maximal \wrt $\LSet[\sttElm]$.
By $\SttSet[\PSym] \subseteq \SttSet$ we denote the set of promotable hybrid
states.
\\ \indent
The main functions of the \emph{hybrid priority-promotion algorithm} (\HPP, for short) reported in
Algorithm~\ref{alg:hybsol} combines the recursive priority-promotion
technique of Algorithm~\ref{alg:expsol} and the recursion-tree truncation idea
of Algorithm~\ref{alg:parsol} in a single approach. Again, the auxiliary functions and procedures are reported in Appendix~\ref{app:c}.
As for the \RPP, the main function $\solFun$ assumes the input state $\sttElm$
to be maximal, \ie, $\sttElm \in \SttSet[\MSym]$.
Line~1 checks whether
\begin{inparaenum}[(i)]
  \item
    there are no unprocessed positions in the game or
  \item
    one of the two bounds on the guarantees over the undetermined positions has
    reached threshold zero.
\end{inparaenum}
If one of these conditions is satisfied, the current region function
$\regFun[\sttElm]$ and the promotion function $\undFun[\sttElm]$ are returned
unmodified at Line~2, as no further progress can be achieved in the current
recursive call.
Otherwise, similarly to Parys' approach, the search for a dominion is split into
three phases:
\begin{inparaenum}[(i)]
  \item
    a first search with halved precision made by calling the auxiliary
    mutually-recursive procedure $\hsolFun$ (Line~3);
  \item
    a second search with full precision via a recursive call to $\solFun$ itself
    (Lines~4 to~8);
  \item
    a final search by means of $\hsolFun$, again with halved precision,
    conditioned to the actual progress obtained during the previous phase
    (Line~9).
\end{inparaenum}
Once these three phases terminate, the information about the undetermined
positions %still 
contained in the local area $\LSet[\sttElm]$ or in the
undetermined set $\USet[\sttElm]$ is suitably updated by function $\UndFun$
at Line~10.
\\ \indent
To discuss the guarantees and their effects in more detail, let us fix a small
dominion $\DSet$, with $\card{\DSet} \leq b_{\dual{\alpha}[\sttElm]}$.
The call to $\hsolFun$ at Line~3 modifies in-place the maximal state $\sttElm$
given as input into a strongly-maximal one such that
$\mathsf{dom}(\regFun^{(\leq p_s)})$ does no longer contain any tiny dominions
of player $\dual{\alpha}[\sttElm]$ of size
$\leq \lfloor \frac{b_{\dual{\alpha}[\sttElm]}}{2} \rfloor$.
Moreover, $\DSet' = \DSet \cap \mathsf{dom}(\regFun^{(\leq p_s)})$ is a dominion
in $\mathsf{dom}(\regFun^{(\leq p_s)})$, while $\DSet \setminus \DSet'$ has
been processed and added to $\mathsf{dom}(\undFun^{(\geq p_s)}) \cup
\mathsf{dom}(\regFun^{(> p_s)})$.
After that, at Line~4, the obtained state is locally recorded in order to
determine, later on, whether the second phase achieves any progress.
The algorithm then proceeds to analyse the remaining part of the game still
unprocessed.
To do so, the next state computed by $\NxtFun(\sttElm)$ is given as input to the
recursive call at Line~5.
Once the call completes, the state is updated with the two new functions
returned by the call.
At this point, the guarantee that all $\dual{\alpha}[\sttElm]$ dominions in
$\mathsf{dom}(\regFun^{(\leq p_s)})$ are larger than $\lfloor
\frac{b_{\dual{\alpha}[\sttElm]}}{2} \rfloor$ entails that, if $\DSet'$ is not
empty (and thus $\DSet$ completely processed), then a non-empty sub-dominion
$\DSet''$ of $\DSet'$ is part of the call.
As $\mathsf{dom}(\regFun^{(\leq p_s)})$ contains no tiny
$\dual{\alpha}[\sttElm]$ dominion, $\card{\DSet''} > \lfloor
\frac{b_{\dual{\alpha}[\sttElm]}}{2}\rfloor$, and
$\DSet''\setminus \DSet' \leq \lfloor
\frac{b_{\dual{\alpha}[\sttElm]}}{2}\rfloor$.
%
%\\ \indent

Depending on whether the state is closed or not, either a promotion or a
maximisation operation is performed (Lines~6 to 8), to ensure that the new state
is maximal.
If the middle phase has made some progress in the search, a last call to
$\hsolFun$ at Line~9 is performed, which again modifies in-place the current
state into a strongly-maximal one.
As $\DSet''\setminus \DSet'\leq \lfloor
\frac{b_{\dual{\alpha}[\sttElm]}}{2}\rfloor$, it is processed in $\hsolFun$.
If no progress occurred, instead, the current state is equal to the one
previously returned by the first call to $\hsolFun$ and, thus, strongly-maximal.
It also entails that $\DSet'$ was empty, and $\DSet$ therefore processed
completely in the fist call of $\hsolFun$.
In both cases, the state is fed to the function $\UndFun$, after which the
current call terminates.

The procedure $\hsolFun$ simply executes the main body of the \RPP algorithm by
making mutually-recursive calls to the $\solFun$ function (Line~5) with halved
precision, until no progress on the search for a dominion can be made.
As for the \RPP, the auxiliary function $\NxtFun$ identifies the next
priority to consider, and the promotion and maximisation procedures, $\PrmFun$ and $\MaxFun$, generalise the corresponding ones associated with \RPP. Hence, $\PrmFun$ applies a promotion while $\MaxFun$ makes a state $\dual{\alpha}$-maximal.
Similarly to Parys' algorithm, the $\HalfFun$ function halves the bound of the
opponent player $\dual{\alpha}[\sttElm]$, and, finally, the $\UndFun$ function, updates or reset the 'undetermined' set of positions.

At this point, by defining the winning regions of the players as $\WSet[\OppSym]
= \regFun^{-1}(\topOSym) \cup \undFun^{-1}(\topZSym)$ and $\WSet[\PlrSym] =
\PosSet \setminus \WSet[\OppSym]$, \ie, $\solFun(\GameName) \defeq
(\WSet[\PlrSym], \WSet[\OppSym])$, where $(\regFun, \undFun) \defeq
\solFun(\isttElm)$, we obtain a sound and complete solution algorithm, whose
time-complexity is quasi-polynomial, as we shall show in the next section.

\vspace{-1.25em}
\noindent
\begin{minipage}[t]{0.496\textwidth}
  \null
  \algsolqsipoll
\end{minipage}
\hspace{-0.5em}
\begin{minipage}[t]{0.504\textwidth}
  \null
  \algsolqsipolhalf
  \label{alg:hybsolhalf}
\end{minipage}

% End of file PrioPromB.tex

%% file: PrioPromC.tex
% Begin of file PrioPromC.tex

\section{Correctness and Complexity}

We now discuss how we can entangle the concepts of Priority Promotion---the
transfer of information across the call structure, which makes it so efficient
in practice---with the concept of relative guarantees that provides favourable
complexity guarantees to Parys' algorithm.
Before turning to the principle guarantees provided by the algorithm, we note
that the two algorithms from the previous sections, Parys' algorithm and the
selected variation of Priority Promotion, can be viewed as variations of our
hybrid algorithm.
This is particularly easy to see for the exponential Priority Promotion algorithm
from Section \ref{sec:hybalg;sub:prtprm}:
when we set the bounds to infinity---or to $2^{c}$, where $c$ is the number of
different priorities of the game---then the algorithm never runs out of bounds.
In this case, the function $u$ is never used, and the algorithm behaves exactly
as Algorithm~\ref{alg:parsol}.
The connection to Parys' algorithm is slightly looser, but essentially it
replaces $\MaxFun$ by closing only $\USet[\sttElm]$ under attractor, and skips
the promotions (through calling $\PrmFun$) altogether.
Note that these changes would not impact the partial correctness argument, while
the remaining parts of the algorithm alone are strong enough to guarantee
progress.

% \subsection{Outline and Added Datastructure}
%\subsection{Correctness}
\label{sec:corr}

As the algorithm is a hybrid one, its correctness proof has both local and
global aspects.
The global guarantees are that the regions stored in $\regFun^{(>p_s)}$ and the
bounded dominions stored in $\undFun^{(\geq p_s)}$ retain their properties in
all function calls.
These properties are not entirely local, and to conveniently reason about the
effect of updates, we use $\BSet[\sttElm] =
\HSet[\sttElm][\dual{\alpha}] \setminus \RSet[\sttElm] = \HSet[\undFun][\alpha,{\prtElm[\sttElm]}] \cup (\HSet[\regFun][\dual{\alpha},{\prtElm[\sttElm]}] \setminus \RSet[\sttElm])$
for the states that are \emph{bad for Player $\alpha_s$} in that they contain the states in $u_s^{-1}(q)$ for all $q \geq p_s$ with $q \equiv_2 \alpha_s$ and all states in $r_s^{-1}(q)$ for all $q > p_s$ with $q \not\equiv_2 \alpha_s$, and
%\bigcup_{q \in \mathsf{rng}(u^{\geq p}), q \equiv_2 p} u^{-1}(q)
%\cup
%\bigcup_{q \in \mathsf{rng}(r^{> p}), q \not\equiv_2 p} r^{-1}(q).$
$\GSet[\sttElm] =
\HSet[\sttElm][\alpha] \setminus \RSet[\sttElm] = \HSet[\undFun][\dual{\alpha},{\prtElm[\sttElm]}] \cup (\HSet[\regFun][\alpha,{\prtElm[\sttElm]}] \setminus \RSet[\sttElm])$
for the states that are \emph{good for Player $\alpha_s$} in that they contain the states in $u_s^{-1}(q)$ for all $q > p_s$ with $q  \not\equiv_2 \alpha_s$ and all states in $r_s^{-1}(q)$ for all $q > p_s$ with $q\equiv_2 \alpha_s$.

We also introduce additional data for each function, namely a set
$\PSet[\sttElm]$, which stores the local area $\LSet[\sttElm]$ at the beginning
of the call of $\solFun$, which is then available also in the $\hsolFun$ at the
level where they are called.
This additional set $\PSet[\sttElm]$ of initial positions is relevant, as the
guarantees of finding small dominions is formulated relative to this initial
set, and not relative to the local area $\LSet[\sttElm]$ at the end of
$\solFun$.
The correctness proof falls into lemmas that refer to the guarantees maintained
by the auxiliary functions, and an inductive proof of the main theorem.
While the proofs for the auxiliary functions and of the main result 
will be reported in the extended version,
%are moved to Appendix~\ref{app:hybrprioprom} and~\ref{app:mainth}, respectively,
the inductive proof of the correctness is outlined below.

% End of file PrioPromC.tex

%% file: Uncert.tex
% Begin of file Uncert.tex

%\subsection{Main Proof}
%\label{sec:mainproof}

\begin{restatable}{theorem}{main}
\label{theo:main}
Let $\sttElm \in \SttSet[\MSym]$ be a maximal hybrid state for a parity game
$\GameName$, where $\cprtElm[\sttElm] = \min \{ \mathsf{dom}(\PSet[\sttElm])
\cup \topZSym \}$ and $b_{0\sttElm},b_{1\sttElm} \geq 1$.
Assume $\solFun$ is called on $\sttElm$ and let $\sttElm' \defeq
((\der{\regFun}, \prtElm'), (\der{\undFun}, \cprtElm',
\bElm[\PlrSym\sttElm], \bElm[\OppSym\sttElm]))$ be the hybrid state, for
$(\der{\regFun}, \der{\undFun}) \defeq \solFun(\sttElm)$, $\prtElm' \defeq
\cprtElm[\sttElm]$, and some $\cprtElm' > \cprtElm_s$.
The following holds:
\begin{itemize}
  \item
    if $\cprtElm[\sttElm] \equiv_{2} \alpha_{\sttElm}$ then $\BSet[\sttElm] =
    \BSet[\sttElm']$, hence, $\BSet[\sttElm']$ preserves all global guarantees
    of $\BSet[\sttElm]$.
  \item
    if $\cprtElm[\sttElm] \not\equiv_{2} \alpha_{\sttElm}$ then:
    \begin{itemize}
      \item
        $\GSet[\sttElm]$ contains all small dominions of
        player $\alpha_{\sttElm}$ of size $\leq \bElm[\alpha_{\sttElm}]$  in $P_s$ and
        intersects with no small $\dual{\alpha}_{\sttElm}$ dominions of size
        $\leq \bElm[\dual{\alpha}_{\sttElm}]$ in $P_s$; and
      \item
        $\BSet[\sttElm'] = \GSet[\sttElm]$.
    \end{itemize}
\end{itemize}
\noindent
Moreover, if $\hsolFun$ is called on $\sttElm \in \SttSet[\MSym]$, then it
modifies $\sttElm$ into a strongly-maximal hybrid state with the following
property: $\LSet[\sttElm]$ does not contain a small dominion of player
$\dual{\alpha}[\sttElm]$ of size $\leq \lfloor
\frac{b_{\dual{\alpha}[\sttElm]}}{2} \rfloor$, and $\RSet[\sttElm]$ and
$\BSet[\sttElm]$ are closed under attractor in $\PSet[\sttElm]$.

%For $S = \mathsf{dom}(r^{\leq p_s}) \cup u^{-1}(p_s)$, $u^{-1}(p_s)$ contains all small dominions of player $1 - \alpha_s$ of size $\leq \lfloor \frac{b_{1-\alpha_s}}{2} \rfloor$, and $R_s$ is closed under attractor.

\end{restatable}
\smallskip

\noindent\textbf{Proof sketch.}
We prove this theorem by induction, using the lemmas from the previous section.
Establishing the base case for $\solFun$ and maximal priority $\bot$, $\hsolFun$
and maximal priority $0$, and $\solFun$ and maximal priority $0$ is straight
forward.

The induction step for $\hsolFun$ then establishes that, after return,
$\LSet[\sttElm]$ contains no tiny $\overline\alpha_s$ dominion of
size $\leq \lfloor \frac{b_{\overline\alpha_s}}{2} \rfloor$, as they would
otherwise (using Lemma \ref{lem:no-p}) be found in the last recursive call of
$\solFun$.
We also have that $\RSet[\sttElm]$ is closed under attractor because,
since closed states lead to promotion, the resulting state $\sttElm$ must be
open, so that $\MaxFun$ provides closure of $\RSet[\sttElm]$ under attractor.

Using this, the induction step for $\solFun$ has mainly to show that an initial
small dominion $\DSet$ ($\card{\DSet} \leq b_{\overline\alpha_s}$) of player
$\overline\alpha_s$ is entirely in $\BSet[\sttElm]$ when $\UndFun$ is called.
Let $\DSet'$ be the intersection of $\DSet$ with $\LSet[\sttElm]$.
Obviously, $\DSet'$ is a dominion of player $\overline\alpha_s$ (Lemma
\ref{lem:noattraction}).
If $\DSet'$ is empty, we are done.
Otherwise, $\DSet'$ must have a non-empty sub-dominion $\DSet''$ that does not
intersect with $\RSet[\sttElm]$ (Lemma \ref{lem:no-p}), and these are added, by
inductive hypothesis, to $\BSet[\sttElm]$ by the full precision call (Line~5).
In addition, $\card{\DSet''}> \lfloor \frac{b_{\overline\alpha_s}}{2} \rfloor$
by the return guarantees of $\hsolFun$.
The rest of the dominion included in $\DSet'\setminus \DSet''$, is then
$\leq \lfloor \frac{b_{\overline\alpha_s}}{2} \rfloor$ and, by the guarantees of
$\hsolFun$, is added to $\BSet[\sttElm]$ by the second call of $\hsolFun$ (note
that we have $\sttElm' \prec \sttElm$).

With the global guarantee that $\BSet[\sttElm]$ does not intersect with small
dominions of size $\leq b_{\alpha_{\sttElm}}$ of player $\alpha_{\sttElm}$, we
have added all small dominions of the players $\overline\alpha_s$ and
$\alpha_s$ to $\BSet[\sttElm]$ and $\GSet[\sttElm] \cup \LSet[\sttElm]$,
respectively and, depending on the priority, we can insert either $U_s$ or
$L_s$ into $u^-1(c_s)$.
\qed

Correctness is then the special case that we call $\solFun$ with $c_s =
\topZSym$ and full precision.

\begin{corollary}
When $\solFun$ is called with the initial state $\isttElm$, \ie with $\cprtElm[\sttElm] = \top_0$ and full precision $\bElm[\alpha]=\card{\GameName}$, for all $\alpha \in \SetB$, then, after $\solFun$ returns, $r^{-1}(\top_1) \cup u^{-1}(\top_0)$ is the winning region of player $1$. \qed
\end{corollary}

This is because $r^{-1}(\top_1)$ and $r^{-1}(\top_0)$ contain dominions of the respective players and are closed under attractor by our global guarantees.
It is also clear that $u^{-1}(\top_0)$ can only be filled in the final call of $\UndFun$, such that
$u^{\top_0}
\cup r^{\top_1}$ contains all dominions (as the bound does not exclude any) of Player $1$, but does not intersect with any dominion (as the bound again does not include any) dominion of Player $0$.
Note that the wining region of Player $0$ is simply the complement of the winning region of Player $1$.
It is interesting to note an algorithmic difference
between the parts of the winning regions in the dominions $r^{-1}(\top_0)$ and
$r^{-1}(\top_1)$, and the rest of the the winning regions of both players:
$r^{-1}(\top_0)$ and
$r^{-1}(\top_1)$, are computed constructively, and winning strategies are simply the contribution attractor strategies / arbitrary exits from states with dominating priority (on their subgame)
%of their region
.
This is not the case for the remainder of the winning regions, as their calculation is not constructive.

%% file: Discussion.tex
% Begin of file Discussion.tex

\section{Discussion}

As a hybrid between Priority Promotion and Parys' approach, the algorithm retains the quasi-polynomial bound of Parys and the practical efficiency of Priority Promotion~\cite{BDM16b,BDM18,BDM18a}.
Our argument is exactly the same as Parys' (Section 5 of \cite{Par19}):
we use 2 parameters, $c$ for the number of priorities, and $l = 2\lfloor \log_2 (n) \rfloor +1$, where $n$ is the number of positions.
We estimate the number of times $\solFun$ is called, excluding the trivial calls that return immediately (in line 2), because we run out of priorities or bounds, by $R(h,l)$.
If $h = 0$, then we have run out of priorities ($p_s= \bot)$, and $R(h,l)=0$. If $l=0$, then we have run out of bounds ($b_{0s} = 0$ or $b_{1s} = 0$, with the other value being $1$).
As argued in Section 5 of \cite{Par19}, we can estimate $R(h,l) \leq n^l \binom{h+l}{l} - 1$.
As the cost of all operations is linear in the size of the game, and as $l$ is logarithmic in the number of positions, this provides a quasi-polynomial running time.

\noindent\textbf{Evaluation and Discussion.}
We have combined two generalisations of the classic recursive algorithm:
the quasi-polynomial recursion scheme of Parys, which relies on the local
spread of imperfect guarantees, and a Priority Promotion scheme, which relies on
identifying and realising the global potential of perfect guarantees.
That these improvements can be synthesised bodes well, as it promises to
perfectly join the advantages of both schemes, and the first experimental data
collected %in Section \ref{sec:expereval}
suggests that the algorithm lives up
to this promise.

%----
The practical effectiveness of the solution algorithms presented here, namely
\RPP and \HPP, has been assessed by means of an extensive experimentation on
both concrete and synthetic benchmarks.
The algorithms have been incorporated in Oink~\cite{Dij18}, a tool written
in \CPP that collects implementations of several parity game solvers proposed
in the literature, including the known quasi-polynomial ones.
We shall compare solution times against the quasi-polynomial solvers
\SSPM~\cite{JL17}, \QPT~\cite{FJSSW17}, and the improved version of Parys' algorithm Par~\cite{LSW19}, as well as the
original version of the best exponential solver classes, namely the recursive
algorithm \ZLK from~\cite{Zie98} %optimized ~\cite{LDT14}
and the original priority promotion
\PP~\cite{BDM18}, whose superiority in practical contexts is widely
acknowledged%
\footnote{Variations of Zielonka's algorithm as well as of Priority Promotion approaches~\cite{BDM16b,BDM18,BDM18a} (including Tangle Learning~\cite{Dij18a} and Justification~\cite{LBD20} based approaches), who share the same basic data structure and promotion principles--are available.
Their performance on benchmarks does not vary significantly.
We therefore went with the original Parity Promotion approach to compare the principle performance.}
(see %quasi-dominion based approaches,
\eg,~\cite{SWW18,Dij18}).
%----

\begin{table}[b]
  \vspace{-1.em}
  \begin{center}
    \small
    \scalebox{0.95}[0.95]{\tabwc}
  \end{center}
  \vspace{-0.5em}
  \caption{\label{tab:wc} \small Solution times in seconds on the worst case family~\cite{BDM16}.}
  \vspace{-2.5em}
\end{table}

The results give a simple argument for why--and when--to use this algorithm\footnote{Experiments were carried out on a 64-bit 3.9 GHz
\textsc{Intel\textregistered} quad-core machine, with i5-6600K processor and
16GB of RAM, running \textsc{Ubuntu}~18.04 with \textsc{Linux} kernel
version~4.15.0. Oink was compiled with gcc~7.4.}.
The first question is why one should use a QP algorithm.
The answer to that question is quite simple: it is not hard to produce pathological cases for exponential time solvers.
For complex games, it is well known that even the most efficient solvers in practice, \ie ZLK~\cite{Fri11a,Dij19,BDM20} and PP~\cite{BDM16,Dij19,BDM20}, take exponential time, while \HPP has a quasi-polynomial worst case complexity.
To show this behaviour, we have evaluated these three solvers and the improved version of Parys’ algorithm Par, which is a quasi polynomial version of ZLK, on the worst case family developed for the approaches based on quasi-dominions~\cite{BDM16}. The results are reported in Table~\ref{tab:wc}. Clearly the complex infrastructure required by the HPP can pay off in terms of running time, while Par does not outperform ZLK on these small examples. We guess it will eventually, but the game size has to grow significantly for this to show.

The reason for why most QP algorithms should not be used in practice is given with Keiran's family of (for current solvers) simple benchmarks~\cite{Kei15} in Appendix~\ref{sec:expereval}: our algorithm takes twice as long as the leading algorithm (\PP) on this benchmark as a whole, the difference is between 0.17 and 0.36 second on average; the classic quasi-polynomial algorithms (\QPT and \SSPM) already time out; and $Par$ fares similar to our algorithm. This is because these benchmarks are quite simple, and neither our algorithm nor $Par$ runs out of bound. We also included the results for the non-bounded recursive version \RPP.

\figperthreeR

We then tried to assess scalability \wrt the
number of positions and priorities, so as to evaluate how sensitive the solvers
are to variations of those two parameters.
To this end, we set up two types of synthetic benchmarks.
The first kind of benchmarks keeps the number of priorities fixed and only
increases the size of the underlying graph, while the second one maintains a
linear relation between positions and priorities. Here we drop both \SSPM and \QPT,
since they could not solve any of these benchmarks.

Figure~\ref{fig:bigper} reports the solution times of the quasi polynomial
solvers on 10 clusters, each composed of 100 randomly generated games, of
increasing size varying from $10^4$ to $10^5$. Each point corresponds to the
total time to solve all the games in the cluster\footnote{The instances were generated by issuing the following OINK command for left-side games: rngame n n/4 2 10; right-side games: rngame n 500 2 10}.
On the left-hand side the number of priorities grows linearly with the
positions, \ie, equal to $\frac{n}{4}$, with $n$ the number of positions.
On the right-hand side, instead, all the games have $500$ priorities.
In both cases, the timeout was set to 25 seconds.
In these experiments, \HPP and $Par$ are tested toghether with the exponential solvers.
%The results seem to confirm what we already observed with the decision procedure
%benchmarks of the previous subsection.
\HPP definitely scales very well \wrt the number of priorities, as opposed to
$Par$, which is very sensitive to this parameter and starts hitting the timeout
already on the smallest instances.
\HPP, indeed, behaves very much like the exponential solvers, none of which
seems to be particularly sensitive to this parameter in practice, despite
requiring time exponential in the number of priorities in the worst case.

What seems to emerge from the experimental analysis is that \HPP behaves quite
nicely in practice, competing with the leading exponential solution
algorithms:
the algorithmic overhead that guarantees its quasi-polynomial upper bound does
not seem to impact the performance in any meaningful way, unlike what happens
for all the other quasi-polynomial solvers, which do \emph{not} scale with the number
of positions and/or the number of priorities.
This bodes well for the applicability of the hybrid approach in more
challenging practical contexts, such as deciding temporal logic properties or
solving reactive synthesis problems, where the number of priorities is typically
higher.

Thus, \HPP should be used when the game is hard; for this, it should have \emph{some} structure (as opposed to be essentially random) as well as a high number of priorities, as the used advanced data structure only kicks in when the number of priorities is higher than $\log(n)$.

The excellent performance of the basic interplay between the two parts of the data structure invites exploring the limits of this approach. 
In future work, we will refine the interplay between these parts in our
algorithm, e.g.\ by extracting the guarantees on the bounds provided upon return instead of the guarantees required when a call is made.

% End of file Discussion.tex

%% file: Acknowledgments.tex
% Begin of file Acknowledgments.tex

\begin{section}*{Acknowledgments}

  \includegraphics[height=8pt]{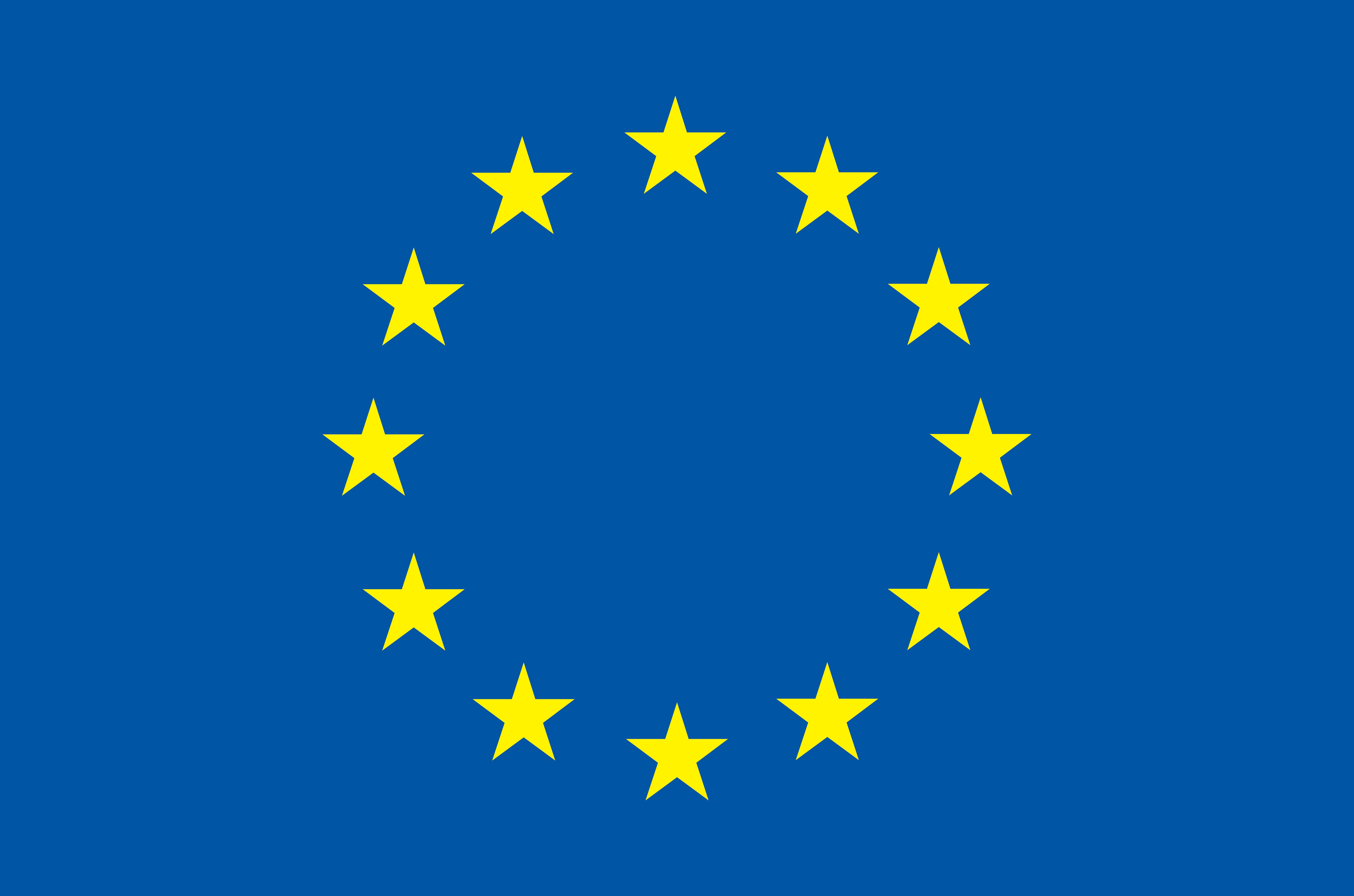} This project has received funding from the European Union’s Horizon 2020 research and innovation programme under grant agreement No 101032464.
  F.~Mogavero acknowledges a partial support by GNCS 2020 project ``Ragionamento
  Strategico e Sintesi Automatica di Sistemi Multi-Agente''.
  The work was supported by EPSRC grant EP/P020909/1.

\end{section}

% End of file Acknowledgments.tex

%% file: AppendixB.tex
% Begin of file AppendixB.tex

\section{Appendix for Section \ref{sub:parys}}
\label{app:b}

\noattraction*
\begin{proof}
From the definition of a dominion, Player $\alpha$ has a strategy that, from
within $D$, only agrees with plays that stay within $D$, contradicting any node
in
$D$ being within the attractor of $S$.
\end{proof}

\nop*
\begin{proof}
Let us fix a strategy, which witnesses that $D$ is a dominion for Player
$\alpha$ in
$\GameName$.
We now assume for contradiction that player $\overline\alpha$, for this
strategy, player $\overline{\alpha}$ can attract the player from any point in
$D$ to $R$. Then, by playing this attractor strategy outside of $R$ and a
witnessing strategy for $R$ being a $q$-region in $R$.
We note that a play also stays in $D$ due to the witness strategy of player
$\alpha$.
An ensuing play is a lasso path, which either eventually stays in $R$ (in which
case player $\overline{\alpha}$ wins due to her wintess strategy), or it
infinitely often enters and leaves $R$, in which case it passes escape positions
of $R$ infinitely often. We would then have that $q$ is the highest priority
that occurs infinitely often, such that player $\overline{\alpha}$ wins.
This contradicts that
$D$ is a dominion.
Thus, there is a well defined non-empty area $A$ in $D$ from which player
$\overline{\alpha}$ cannot attract to $R$. $A$ is a dominion, as player $\alpha$
can win on $A$ with the same witness strategy.
As $A$ clearly does not intersect with $R$, it does not intersect with
$\atrFun_{\GameName}^{\overline{\alpha}}(R)$ either By Lemma
\ref{lem:noattraction}.
\qed
\end{proof}
Having found such a dominion it can be closed under attrator and, after removing
this attractor, what remains from a dominion is still a dominion, which allows
for stepwise collecting dominions.
\begin{lemma}\label{lem:attract}
Let $D$ be a dominion for Player $\alpha$ in a game $\GameName$. Then, for $A =
\atrFun_{\GameName}^\alpha(D)$, then $A$ is a
dominion for Player $\alpha$ in $\GameName \setminus  A$.
\end{lemma}
\begin{proof}
It suffices to use a witness strategy for $D$ being a dominion in $D$, and an
attractor strategy to $D$ in the remainder of $A$.
\qed
\end{proof}
\begin{lemma}\label{lem:coattract}
Let $D$ be a dominion for Player $\alpha$ in a game $\GameName$. Then, for all
sets $S$, if $A = \atrFun_{\GameName}^{\overline\alpha}(S)$, then $D\setminus A$
is a
dominion for Player $\GameName$ in $\GameName \setminus  A$.
\end{lemma}
\begin{proof}
The same strategy that witnesses $D$ being a dominion for Player $\alpha$ in
$\GameName$ witnesses $D\setminus A$ being a dominion for Player $\alpha$ in
$G\setminus A$.
\qed
\end{proof}
Together, this implies that \textsf{hsol} returns in a situation, where there
are no small dominions left outside of $\USet[\sttElm]$.
Broadly speaking, this means that, if the full precision call of \textsf{sol} in
\textsf{sol} cuts a chunk $> \lfloor \frac{b_{\overline \alpha_s}}{2} \rfloor$
from any remaining dominion $\leq b_{\overline \alpha_s}$, leaving a remainder
of that dominion of size $\leq \lfloor \frac{b_{\overline \alpha_s}}{2}
\rfloor$, which is then found by the second call of \textsf{hsol}.
After Parys' Solver is run (with full precision) for a parity game with maximal
priority $p_{\max}$, $u^{-1}(p_{\max}+1)$ contains the winning region of player
$p_{\max} \mod 2$.

% End of file AppendixB.tex

%% file: AppendixA.tex
% Begin of file AppendixA.tex

\section{Properties of the auxiliary functions for Section~\ref{sec:corr}}
\label{app:hybrprioprom}

\begin{restatable}{lemma}{lmmnxthalf}
  \label{lmm:nxthalf}
  Given a strongly-maximal hybrid state $\sttElm \in \SttSet[\SSym]$, the
  functions $\NxtFun(\sttElm)$ and $\HalfFun(\sttElm)$ return a maximal hybrid
  state $\der{\sttElm} \in \SttSet[\MSym]$ such that $\der{\sttElm} \prec
  \sttElm$.
\end{restatable}
\begin{proof}
  Let $\sttElm \defeq ((\regFun, \prtElm), (\undFun, \cprtElm, \bElm[\PlrSym],
      \bElm[\OppSym]))$ be a strongly-maximal hybrid state. For function $\NxtFun$ let be $\der{\sttElm} = \NxtFun(\sttElm)$ the result obtained by computing the function
$\NxtFun$ on $\sttElm$. Due to Line 2 of Algorithm $\NxtFun$, $\der{\sttElm} = ((\regFun[\sttElm], \qprtElm), (\undFun[\sttElm], \prtElm[\sttElm], \bElm[0\sttElm], \bElm[1\sttElm]))$ where $\qprtElm = \max(\rng{\regFun[\sttElm]^{(< \prtElm[\sttElm] )}})$ by Line 1.
Consequently, trivially holds Condition~\ref{def:hybsrcspc(par)} of Definition~\ref{def:hybsrcspc} and that $\regFun[\sttElm]$ is a region function, $\undFun[\sttElm]$ is a promotion function, $\prtElm[\sttElm]$ is a priority, and
$\bElm[\alpha]$ for $\alpha\in\SetB$ are two integers.
 Moreover, clearly $\qprtElm\in[\bot,{\prtElm[\sttElm]}[$ and, therefore, also $\qprtElm$ is a priority and $\undFun[][-1](\topOSym) = \dom{\undFun[][(<\qprtElm)]} = \emptyfun$ as required by Condition~\ref{def:hybsrcspc(und)} of the same definition. In addition, since $\HSet[\sttElm][\beta, \qprtElm] = \HSet[\der{\sttElm}][\beta, {\prtElm[\sttElm]}] \cup \regFun[][-1](\qprtElm)$ and $\HSet[\der{\sttElm}][\dual{\beta}, \qprtElm] = \HSet[\sttElm][\dual{\beta}, {\prtElm[\sttElm]}]$, with $\beta = \qprtElm \bmod{2}$, it also follows Condition~\ref{def:hybsrcspc(esc)}. Finally, to prove that $\der{\sElm}$ is maximal it suffice to observe that $\HSet[\der{\sttElm}][\alpha] = \HSet[\sttElm][\alpha] \cup \RSet[\sttElm]$ and $\HSet[\der{\sttElm}][\dual{\alpha}] = \HSet[\sttElm][\dual{\alpha}]$. Consequently, since $\sttElm$ is strongly-maximal it holds that $\RSet[\sttElm]$ is $\alpha$-maximal \wrt $\LSet[\sttElm]$ and, therefore, $\HSet[\der{\sttElm}][\alpha]$ is $\alpha$-maximal as well. To conclude, if $\RSet[\sttElm]\neq\emptyset$ trivially we have that $\LSet[\der{\sttElm}]\subset\LSet[\sttElm]$, otherwise we have that $\LSet[\der{\sttElm}] = \LSet[\sttElm]$ and $\prtElm[\der{\sttElm}] = \qprtElm < \prtElm[\sttElm]$. In both case we can conclude that $\der{\sttElm} \prec \sttElm$.

 We can now focus on $\HalfFun$ function. Let us consider $\der{\sttElm}$ as the argument of function $\NxtFun$ at Line 2 of Algorithm $\HalfFun$. Due to Line 1, we have that $\der{\sttElm} =
    ((\regFun, \prtElm), (\undFun, \cprtElm, \der{\bElm[\PlrSym]}, \der{\bElm[\OppSym]}))$, where $\der{\bElm[\PlrSym]} = \lfloor\frac{\bElm[0\sElm]}{1 + \alpha}\rfloor$ and $\der{\bElm[\OppSym]} = \lfloor\frac{\bElm[1\sElm]}{2 - \alpha}\rfloor$. It is easy to observe, by the fact that $\sttElm$ is a strongly-maximal hybrid state and that both $\der{\bElm[\PlrSym]}$ and $\der{\bElm[\OppSym]}$ are integers, that $\der{\sttElm}$ is a strongly-maximal hybrid state as well. Hence, $\HalfFun$ returns $\NxtFun(\der{\sttElm})$ that in turn returns a new maximal hybrid
  state $\sttElm' \in \SttSet[\MSym]$ such that $\sttElm' \prec \der{\sttElm}$.
\end{proof}

\begin{restatable}{lemma}{lmmmax}
  \label{lmm:max}
  Given a hybrid state $\sttElm \in \SttSet$, the procedure $\MaxFun(\sttElm)$
  modifies in-place $\sttElm$ into a maximal hybrid state $\der{\sttElm} \in
  \SttSet[\MSym]$ such that $\der{\sttElm} \preccurlyeq \sttElm$.
\end{restatable}

\begin{proof}
Let $\sttElm \defeq ((\regFun, \prtElm), (\undFun, \cprtElm, \bElm[\PlrSym],
      \bElm[\OppSym]))$ be an hybrid state and
$\der{\sttElm} \defeq
    ((\der{\regFun}, \der{\prtElm}), (\der{\undFun}, \der{\cprtElm}, \der{\bElm[\PlrSym]}, \der{\bElm[\OppSym]}))
= \MaxFun[][](\sttElm)$ the modified state obtained by computing the function
$\MaxFun$ on $\sttElm$.
It is easy to see that $\der{\prtElm} = \prtElm$, $\der{\cprtElm} = \cprtElm$, and
$\der{\bElm[\alpha]}=\bElm[\alpha]$ for $\alpha\in\SetB$.
Moreover, due to Line 3 of Algorithm $\MaxFun$, it holds that $\XSet \subseteq \RSet[\sttElm] \cup \LSet[\sttElm] \subseteq\dom{\prtFun[][(\leq \prtElm)]}$ while, due to Line 4, we have that $\qElm\geq\prtElm$.
As a consequence, by Lines 6 and 8, $\der{\regFun}$ is a a promotion function.
Furthermore, from the observation that in case $\qElm = \top_{\beta}$, with $\beta\in\SetB$, it necessarily holds that $\alpha \equiv_2 \beta$, we can also conclude that, by Lines 7 and 9, $\der{\undFun}$ is a promotion function as well.
To prove that $\der{\regFun}$ is a region function, let us consider again the case in which $\qElm = \top_{\beta}$ with $\beta\in\SetB$. Now, since $\alpha \equiv_2 \beta$, the set $\XSet$ corresponds to the $\beta$-attractor of $\RSet[\sttElm] \cup \LSet[\sttElm]$ by Line 3.
Once $\XSet$ is merged to $\regFun[][-1](\top_{\beta})$, by Line 6, it is obvious that $\regFun[][-1](\top_{\beta})$ is still a $\beta$-dominion due to the fact that the $\beta$-attractor of a $\beta$-dominion is a $\beta$-dominion as well.
We can now consider $\qElm < \top_{\beta}$ for any $\beta\in\SetB$.
Let us first focus on the interation of Line 2 in which $\alpha = \alpha_{\sttElm}$.
Two cases may arises: $\qElm \equiv_2 \alpha$ or $\qElm \not\equiv_2 \alpha$.
In the first case, due to Lines 5-6, we have that $\HSet[\der{\regFun}][\alpha] = \HSet[\regFun][\alpha] \cup \XSet$ which is clearly a quasi $\alpha$-dominion since $\XSet$ corresponds to the $\alpha$-attractor of $\RSet[\sttElm] \cup \LSet[\sttElm]$, while $\HSet[\der{\regFun}][\dual{\alpha}] = \HSet[\regFun][\dual{\alpha}]$ is a quasi $\dual{\alpha}$-dominion. Hence, $\der{\regFun}$ is a quasi dominion function.
Finally, clearly $\XSet$ does not belongs to $\escFun[][\dual{\alpha}](\HSet[\der{\regFun}][\alpha])$, since Player $\alpha$ has a strategy to attract every position $\vElm\in\XSet$ to $\regFun[][-1](\qElm)$, form which the opponent can only escape through positions having priority $\qElm'\geq\qElm$ and congruent to the parity of $\qElm$.
Summing up, the resulting $\der{\regFun}$ is a region function and due to the fact that the priority function $\prtFun$ is always a region function, we can conclude that also after Line 10 $\der{\regFun}$ is a region function.
On the other hand, if $\qElm \not\equiv_2 \alpha$, we can easily conclude that $\der{\regFun}$ is a region function due to Lines 5, 8, and 10.
When, instead, the loop at Line 2 selects a parity $\alpha \neq \alpha_{\sttElm}$ and the priority selected at Line 4 has the same parity as $\alpha$, the proof that $\der{\regFun}$ is a region function is equivalent to the one provided above for the case in which $\qElm \not\equiv_2 \alpha$ and $\alpha = \alpha_{\sttElm}$.
Similarly, when $\qElm \not\equiv_2 \alpha$ the proof corresponds to the one provided for the case in which $\qElm \equiv_2 \alpha$ and $\alpha = \alpha_{\sttElm}$.
To prove Condition~\ref{def:hybsrcspc(par)} of Definition~\ref{def:hybsrcspc} it suffice to observe that whenever $\XSet$ is merged to $\regFun$ (resp. $\undFun$), it is removed from $\undFun$ (resp. $\undFun$) due to Lines 6-7 and 8-9, while the operation at Line 10 does not change $\dom{\regFun}$, since $\LSet[\sttElm] \defeq \dom{\regFun[][(\leq \prtElm)]}$.
Condition~\ref{def:hybsrcspc(und)}, instead, follows form the fact that that $\qElm\geq\prtElm$ and $\qElm = \top_{\beta}$ with $\beta\in\SetB$ only if $\alpha \equiv_2 \beta$.
Therefore, it holds that $\der{\undFun}[][-1](\topOSym) = \undFun[][-1](\topOSym)$, and $\dom{\der{\undFun}[][(<\prtElm)]} = \dom{\undFun[][(<\prtElm)]}$ which are all empty.
It remains to prove Condition~\ref{def:hybsrcspc(esc)}.
Since $\sttElm$ is a hybrid state, it holds that $\USet[\sttElm]$ and every $\vElm \in \HSet[\undFun][\dual{\alpha}]$ cannot be attracted by $\RSet[\sttElm]$ and more in general by $\LSet[\sttElm]$, due to Condition~\ref{def:hybsrcspc(esc)}.
Due to the same condition, it also holds that $\dom{\undFun} \cap \escFun[][\dual{\alpha}](\HSet[\sttElm][\alpha, \qElm']) = \emptyset$ for every $\qElm'>\prtElm$.
In addition, by Lines 5 and 9, the domain of $\undFun$ only increase when the priority selected at Line 4 has parity $\dual{\alpha}$.
However, the positions of $\XSet$ merged in this case are not part of $\escFun[][\dual{\alpha}](\HSet[\der{\sttElm}][\alpha])$ since $\XSet$ has been computed as the $\dual{\alpha}$-attractor to $\HSet[\undFun][\dual{\alpha}]\setminus \LSet[\sttElm]$ and, therefore, Player $\dual{\alpha}$ has a strategy to reach $\HSet[\undFun][\dual{\alpha}]\setminus \LSet[\sttElm]$ from $\XSet$.
The maximality of $\der{\sttElm}$ is a trivial consequence of Line 3, together with the observation that the operation at Line 10 cannot affect the maximality of the state.
Finally, since $\dom{\der{\undFun}} \supseteq \dom{\undFun}$ and $\dom{\der{\regFun}} \supseteq \dom{\regFun}$, while $\der{\prtElm} = \prtElm$, we can also conclude that $\der{\sttElm} \preccurlyeq \sttElm$.
\end{proof}

\begin{restatable}{lemma}{lmmprm}
  \label{lmm:prm}
  Given a promotable hybrid state $\sttElm \in \SttSet[\PSym]$, the procedure
  $\PrmFun(\sttElm)$ modifies in-place $\sttElm$ into a maximal hybrid state
  $\der{\sttElm} \in \SttSet[\MSym]$ such that $\der{\sttElm} \prec \sttElm$.
\end{restatable}

\begin{proof}
Let $\sttElm \defeq ((\regFun, \prtElm), (\undFun, \cprtElm, \bElm[\PlrSym],
      \bElm[\OppSym]))$ be a promotable hybrid state and
$\der{\sttElm} \defeq
    ((\der{\regFun}, \der{\prtElm}), (\der{\undFun}, \der{\cprtElm}, \der{\bElm[\PlrSym]}, \der{\bElm[\OppSym]}))
= \PrmFun[][](\sttElm)$ the modified state obtained by computing the function
$\PrmFun$ on $\sttElm$.
It is easy to see that $\der{\prtElm} = \prtElm$, $\der{\cprtElm} = \cprtElm$, and
$\der{\bElm[\alpha]}=\bElm[\alpha]$ for $\alpha\in\SetB$.
At Line 1 of Algorithm $\PrmFun$ the
two best escape priorities $(\prtElm[\regFun], \prtElm[\undFun])$ from $\RSet[\sttElm]$ to $\rFun$ and $\uFun$ are computed by the function $\bepFun$.
Due to the
definition of $\bepFun$ and the fact that $\RSet[\sttElm]$ is an $\alpha$-dominion in
$\GameName[\sElm]$, it follows that both $\prtElm[\regFun]$ and $\prtElm[\undFun]$ are
priorities higher than $\prtElm$.
Moreover, by the same definition of $\bepFun$ and the fact that $\sttElm$ is promotable, it also holds that $\prtElm[\regFun]\equiv_{2}\prtElm$, while $\prtElm[\undFun]\not\equiv_{2}\prtElm$ if
$\prtElm[\undFun]<\top_{\alpha_{\sttElm}}$.
Now, two cases may arises: $\prtElm[\regFun]\leq\prtElm[\undFun]$ or $\prtElm[\regFun] >
\prtElm[\undFun]$.
In the fist case, by Line 2 and 3, we have that
$\der{\undFun}=\undFun$ and $\der{\regFun}={\regFun}[\RSet[\sttElm] \mapsto{\prtElm[\regFun]}]$.
Hence, $\der{\undFun}$ is a promotion function and Condition~\ref{def:hybsrcspc(und)} of Definition~\ref{def:hybsrcspc} is satisfied.
Moreover, since $\RSet[\sttElm]\subseteq\regFun[][(\leq\prtElm)]$, we have that $\dom{\der{\regFun}} =\dom{\regFun}$ and, therefore, also the two requirements of Condition~\ref{def:hybsrcspc(par)} are guaranteed.
Finally, since $\prtElm[\regFun]\equiv_{2}\prtElm$, it follows that $\HSet[\der{\sttElm}][\beta] =\HSet[\sttElm][\beta]$ for any $\beta\in\SetB$.
Hence, Condition~\ref{def:hybsrcspc(esc)} is satisfied.
At this point, it remains to prove that $\der{\regFun}$ is a region function.
Now, since $\RSet[\sttElm]\subseteq\regFun[][(\leq\prtElm)]$ and $\prtElm[\regFun]>\prtElm$, $\der{\regFun}$ is clearly a promotion function.
In addition, it is also a quasi dominion function as a consequence of the fact that $\HSet[\der{\sttElm}][\beta] =\HSet[\sttElm][\beta]$ for any $\beta\in\SetB$.
Moreover, being $\regFun$ a region function, it holds that $\prtFun(\vElm)\geq\prtElm$ and $\prtFun(\vElm)\equiv_{2} \alpha$ for each position $\vElm$ in $\escFun[][\dual{\alpha}](\HSet[\regFun][\alpha,\prtElm])$.
However, due to Line 1 we have that ${\bepFun[\GameName][\dual{\alpha}][{\RSet[\sttElm]},
{\regFun[\sElm]}]}=\prtElm[\regFun]$, which implies that $\prtFun(\vElm)\geq\prtElm[\regFun]$ and $\prtFun(\vElm)\equiv_{2} \alpha$ for each position $\vElm$ in $\escFun[][\dual{\alpha}](\HSet[\regFun][\alpha,{\prtElm[\regFun]}])$. Hence, $\der{\regFun}$ is a region function.
Finally, since $\sElm$ is promotable, it follows that $\HSet[\der{\sttElm}][\alpha]=\HSet[\sttElm][\alpha]\cup\RSet[\sttElm]$ is $\alpha$-maximal, while $\HSet[\der{\sttElm}][\dual{\alpha}]=\HSet[\sttElm][\dual{\alpha}]$ is $\dual{\alpha}$-maximal and, therefore, $\der{\sttElm}$ is maximal.
To conclude, being $\sElm$ promotable, it necessarily holds that $\RSet[\sttElm]\neq\emptyset$, which implies that $\LSet[\der{\sttElm}]\subset\LSet[\sttElm]$.
Hence, $\der{\sttElm} \prec \sttElm$.
Let us now consider the case in which $\prtElm[\regFun] > \prtElm[\undFun]$.
Due to Line
2 and 4 of Algorithm $\PrmFun$ we have that $\der{\regFun}=\regFun \setminus
\RSet[\sttElm]$ and $\der{\undFun}={\undFun}[\RSet[\sttElm] \mapsto {\prtElm[\undFun]}]$.
Conditions~\ref{def:hybsrcspc(par)} and~\ref{def:hybsrcspc(und)} follow from the fact that $\top_{\alpha_{\sttElm}} > \prtElm[\undFun] > \prtElm$ and $\RSet[\sttElm] \subseteq \regFun[][(\leq\prtElm)]$, which also imply that $\der{\undFun}$ is a promotion function.
In addition, since $\RSet[\sttElm] \defeq \regFun[][-1](\prtElm)$, it holds that $\HSet[\der{\regFun}][\alpha,\prtElm]=\HSet[\regFun][\alpha,\prtElm]\setminus\RSet[\sttElm]$, hence
$\der{\regFun}$ is a region function.
To prove Condition~\ref{def:hybsrcspc(esc)}, instead, let us observe that $\RSet[\sttElm]$ is an
$\alpha$-dominion in $\GameName[\sttElm]$ and, consequently,
$\escFun[][\alpha](\RSet[\sttElm])\subseteq\HSet[\sttElm][\alpha]$.
Moreover, we have that
$\prtElm[\undFun]\not\equiv_{2}\prtElm$, hence, as a consequence of the fact that $\der{\undFun}={\undFun}[\RSet[\sttElm] \mapsto {\prtElm[\undFun]}]$, it holds that
$\RSet[\sttElm] \in \HSet[\der{\undFun}][\dual{\alpha},\prtElm] \subseteq \HSet[\der{\sttElm}][\dual{\alpha},\prtElm]$.
In addition, it suffice to observe that
$\dom{\der{\regFun}[][(>\prtElm)]} =\dom{\regFun[][(>\prtElm)]}$ and
$\dom{\der{\undFun}[][(\geq\prtElm)]} =\dom{\undFun[][(\geq\prtElm)]} \cup \RSet[\sttElm]$, with
$\RSet[\sttElm]\neq\emptyset$, to conclude that $\der{\sttElm} \prec \sttElm$.
Finally, similarly to the previous case, it can be proved that both
$\HSet[\der{\sttElm}][\alpha]=\HSet[\sttElm][\alpha]\cup\RSet[\sttElm]$ is $\alpha$-maximal, while $\HSet[\der{\sttElm}][\dual{\alpha}]=\HSet[\sttElm][\dual{\alpha}]$ is $\dual{\alpha}$-maximal, as a direct consequence of the fact that $\sElm$ is promotable.
\end{proof}

\begin{restatable}{lemma}{lmmund}
  \label{lmm:und}
  Given a strongly-maximal hybrid state $\sttElm \in \SttSet[\SSym]$, the
  function $\UndFun(\sttElm)$ returns a pair $(\der{\regFun}, \der{\undFun}) \in
  \RegSet \times \PrmSet$ of a region function $\der{\regFun}$ and a promotion
  function $\der{\undFun}$ such that:
  \begin{inparaenum}[1)]
    \item
      $\regFun[\sttElm][(> {\prtElm[\sttElm]})] \subseteq \der{\regFun}[][(>
      {\prtElm[\sttElm]})]$;
    \item
      $\undFun[\sttElm][(> {\prtElm[\sttElm]})] \subseteq \der{\undFun}[][(>
      {\prtElm[\sttElm]})]$;
    \item
      $\der{\sttElm} = ((\der{\regFun}, \cprtElm[\sttElm]), (\der{\undFun}, \qprtElm, \bSym[0],
      \bSym[1]))$ is a promotable hybrid state if $\der{\sttElm}$ is \emph{closed} and, an hybrid state otherwise, for all priority $\qprtElm \in \prtFun$ such that $\cprtElm[\sttElm] < \qprtElm \leq \min(\rng{{\regFun[\sttElm] \cup \undFun[\sttElm]}^{>\cprtElm[\sttElm]}})$, and $\bSym[0], \bSym[1] \in \SetN$.
  \end{inparaenum}
\end{restatable}

\begin{proof}
Let be $\sttElm \defeq ((\regFun, \prtElm), (\undFun, \cprtElm, \bElm[\PlrSym],
      \bElm[\OppSym]))$ a strongly-maximal hybrid state and $(\der{\regFun}, \der{\undFun})$ the result obtained by computing the function
$\UndFun$ on $\sttElm$.
Moreover, let $\der{\sttElm} = ((\der{\regFun}, \cprtElm[\sttElm]), (\der{\undFun}, \qprtElm, \bSym[0],
      \bSym[1]))$ the composition of state $\sttElm$ and the pair $(\der{\regFun}, \der{\undFun})$.
Two cases may arise: $\cprtElm \equiv_{2} \alpha$ or $\cprtElm \equiv_{2} \dual{\alpha}_{\sttElm}$.
In the first case, by Lines 1-2 of
Algorithm $\UndFun$, we have that $\der{\regFun} = \regFun$ and
$\der{\undFun}={\undFun}[\USet[\sttElm] \mapsto \cprtElm]$, where $\USet \defeq
\undFun[][-1](\prtElm)$.
By the fact that the priorities of the states are strictly decreasing we have that $\cprtElm > \prtElm$. Moreover, it holds that $\dom{\der{\undFun}}=\dom{\undFun}$ and, therefore, $\der{\regFun}$ is a region function, $\der{\undFun}$ is a promotion function, both Points 1) and 2) of the Lemma and both Conditions~\ref{def:hybsrcspc(par)} and~\ref{def:hybsrcspc(und)} of Definition~\ref{def:hybsrcspc} are satisfied for the composed state $\der{\sttElm}$.
To prove Condition~\ref{def:hybsrcspc(esc)} it suffice to
observe that due to Line 3 and the fact that $\cprtElm \equiv_{2} \alpha_{\sttElm}$, it
holds that $\HSet[\der{\sttElm}][\beta] = \HSet[\sttElm][\beta]$ for any $\beta\in\SetB$.
In particular, we have that $\HSet[\der{\regFun}][\alpha_{\sttElm},\qElm] = \HSet[\regFun][\alpha_{\sttElm},\qElm]$ for all priority $\qElm \geq \prtElm$, $\HSet[\der{\undFun}][\dual{\alpha}_{\sttElm},\qElm] = \HSet[\undFun][\dual{\alpha}_{\sttElm},\qElm]$ for all priority $\qElm > \cprtElm$ and $\HSet[\der{\undFun}][\dual{\alpha}_{\sttElm},\cprtElm] = \HSet[\undFun][\dual{\alpha}_{\sttElm},\cprtElm] \cup \USet[\sttElm]$.
Finally, if $\der{\sttElm}$ is \emph{closed}, it necessarily holds that $\der{\sttElm}$ is $\dual{\alpha}_{\sttElm}$-maximal and that $\RSet[\der{\sttElm}]$ is $\alpha$-maximal as a direct consequence of the fact that $\sttElm$ is strongly-maximal and $\regFun[][(>\prtElm)\wedge(<{\cprtElm[\sttElm]})]=\emptyset$.
Hence, $\der{\sttElm}$ is also promotable.
We can now focus on the case in which $\cprtElm \equiv_{2} \dual{\alpha}_{\sttElm}$.
By
Lines 2 and 4-5 of Algorithm $\UndFun$ we have that $\der{\regFun} = {\regFun[]^{( \geq \cprtElm)}}[\posElm \in \USet[\sttElm] \mapsto \prtFun(\posElm)]$ and $\der{\undFun} =
\undFun^{(\geq \cprtElm)}{[\LSet[\sttElm] \mapsto \cprtElm]}$.
Now, it is easy to see that also in this case both Points 1) and 2) of the Lemma and both Conditions~\ref{def:hybsrcspc(par)} and~\ref{def:hybsrcspc(und)} are satisfied for the composed state $\der{\sttElm}$, since $\LSet \defeq \dom{\regFun[][\leq \prtElm]}$ and $\USet\defeq \undFun[][-1](\prtElm)$ and, it also holds that $\der{\undFun}$ is a promotion function and $\der{\regFun}$ is a region function due to the fact that the priority function $\prtFun$ is always a region function.
Condition~\ref{def:hybsrcspc(esc)}, instead, follows form the fact that $\sttElm$ is a strongly-maximal
hybrid state and, therefore, it holds that $\escFun[][\alpha]
(\LSet[\sttElm]) \subseteq \HSet[\sttElm][\dual{\alpha}_{\sttElm}]$. By the same
Condition~\ref{def:hybsrcspc(esc)}, we have that
$\escFun[][\alpha_{\sttElm}](\HSet[\undFun][\alpha_{\sttElm},\prtElm])
\subseteq \preFun[][\dual{\alpha}_{\sttElm}](\HSet[\regFun][\alpha_{\sttElm},\prtElm] \setminus \LSet[\sttElm])$ and, consequently, for each priority $\qElm>\prtElm$ the set
$\undFun[][-1](\qElm)$ cannot reach $\USet[\sttElm]$ that, due to Line 5, is reset to the values of the priority function
$\prtFun$. The latter observation also implies that, if $\der{\sttElm}$ is \emph{closed}, $\der{\sttElm}$ is also $\dual{\alpha}_{\sttElm}$-maximal. Indeed, when $\der{\sttElm}$ is \emph{closed}, it holds that player $\dual{\alpha}_{\sttElm}$ cannot attracts $\RSet[\der{\sttElm}]$ and since $\LSet[\der{\sttElm}] = \RSet[\der{\sttElm}] \cup \USet[\sttElm]$ we can conclude that $\HSet[\der{\sttElm}][\alpha_{\sttElm}]\setminus\LSet[\der{\sttElm}]$ is $\dual{\alpha}_{\sttElm}$-maximal \wrt $\LSet[\der{\sttElm}]$.
Finally, similarly the the previous case, $\RSet[\der{\sttElm}]$ turns out to be $\alpha$-maximal as a direct consequence of the fact that $\sttElm$ is strongly-maximal and that $\regFun[][(>\prtElm)\wedge(<{\cprtElm[\sttElm]})]=\emptyset$.
Hence, $\der{\sttElm}$ is also promotable.
\end{proof}

\begin{restatable}{lemma}{lmmsol}
  \label{lmm:sol}
  Given a maximal hybrid state $\sttElm \in \SttSet[\MSym]$, the function
  $\solFun(\sttElm)$ terminates and returns a pair $(\der{\regFun},
  \der{\undFun}) \in \RegSet \times \PrmSet$ of a region function
  $\der{\regFun}$ and a promotion function $\der{\undFun}$ such that:
 \begin{inparaenum}[1)]
    \item
      $\regFun[\sttElm][(> {\prtElm[\sttElm]})] \subseteq \der{\regFun}[][(>
      {\prtElm[\sttElm]})]$;
    \item
      $\undFun[\sttElm][(> {\prtElm[\sttElm]})] \subseteq \der{\undFun}[][(>
      {\prtElm[\sttElm]})]$;
    \item
      $\der{\sttElm} = ((\der{\regFun}, \cprtElm[\sttElm]), (\der{\undFun}, \qprtElm, \bSym[0],
      \bSym[1]))$ is a promotable hybrid state if $\der{\sttElm}$ is \emph{closed} and, an hybrid state otherwise, for all priority $\qprtElm \in \prtFun$ such that $\cprtElm[\sttElm] < \qprtElm \leq \min(\rng{{\regFun[\sttElm] \cup \undFun[\sttElm]}^{>\cprtElm[\sttElm]}})$, and $\bSym[0], \bSym[1] \in \SetN$.
  \end{inparaenum}
  Moreover, the procedure $\hsolFun(\sttElm)$ modifies in-place $\sttElm$ into a
  strongly-maximal hybrid state $\der{\sttElm} \in \SttSet[\SSym]$ such that
  $\der{\sttElm} \preccurlyeq \sttElm$.
\end{restatable}

\begin{proof}
Let be $\der{\sttElm} \defeq \solFun(\sttElm)$ the result obtained by
computing $\solFun$ on a maximal hybrid state $\sttElm \defeq ((\regFun, \prtElm),
(\undFun, \cprtElm, \bElm[\PlrSym], \bElm[\OppSym]))$.
The proof proceeds by induction, where in the base case we have that $\prtElm=\bot$ or $\bElm[\beta]= 0$ for any $\beta\in\SetB$. In this case, by Line
1 and 2 of Algorithm~\ref{alg:hybsol} we have that $\der{\sttElm}=((\regFun, \cprtElm[\sttElm]), (\undFun, \qprtElm, \bSym[0], \bSym[1]))$, which trivially satisfies all the requirements.
Therefore, let us now consider an arbitrary recursive call of $\solFun$ where $\prtElm>\bot$ and $\bElm[\beta] > 0$ for any $\beta\in\SetB$.
At to Line 3 of $\solFun$, the procedure $\hsolFun$ is called on the state $\sttElm$.
Consequently, at Line 1 of
Algorithm~\ref{alg:hybsolhalf}, the set $\RSet[\sttElm]$ is maximised by computing its $\alpha_{\sttElm}$-attractor $\atrFun[{\GameName_{\sttElm}} ][\alpha_{\sttElm}]$. It is not hard to show that the new state $\sttElm$ is still a strongly-maximal hybrid greater or equal to the old one, \wrt the ordering of Definition~\ref{def:hybsrcspc}.
Then, two cases may arises: $\sttElm$ is \emph{open} or \emph{closed}.
In the first case, at Line 5 a recursive call of $\solFun$ is
performed on the state generated by $\HalfFun$.
Now, due
to the fact that $\sttElm$ is a strongly maximal hybrid state and by
Lemma~\ref{lmm:nxthalf}, it follows that the resulting state of $\HalfFun(\sttElm)$
is a maximal hybrid state greater or equal to the input $\sttElm$, \wrt the ordering of
Definition~\ref{def:hybsrcspc}.
Hence, by external induction, we
have that $\sttElm = ((\der{\regFun}, \prtElm), (\der{\undFun}, \cprtElm, \bElm[\PlrSym], \bElm[\OppSym]))$ is a hybrid state and, by Point 1) and 2) it is also greater or equal to the previous $\sttElm$, \wrt the ordering of
Definition~\ref{def:hybsrcspc}.
At this point, $\sttElm$ can be \emph{open} or \emph{closed}.
Let us first consider again the case in which it is \emph{open}.
By Lemma~\ref{lmm:max} and Line 7, the function $\MaxFun$ modifies $\sttElm$ into a maximal hybrid state that is greater or equal to the previous $\sttElm$, \wrt the ordering of
Definition~\ref{def:hybsrcspc}.
On the other hand,
when $\sttElm$ is \emph{closed}, by external induction it holds that
$\sttElm$ is promotable and, by Lemma~\ref{lmm:prm} the function $\PrmFun$ modifies $\sttElm$ into a
maximal hybrid state greater than the previous one the ordering of
Definition~\ref{def:hybsrcspc}.
This last case also apply when $\sttElm$ is \emph{closed} at Line 4.
Now, since the modified state by $\PrmFun(\sttElm)$ is always greater than $\sttElm$, \wrt the ordering of
Definition~\ref{def:hybsrcspc}, we can conclude that function $\hsolFun$ only ends when the
function $\MaxFun$ modify the input hybrid state such that there is no progress \wrt the ordering of
Definition~\ref{def:hybsrcspc}, otherwise $\hsolFun$ starts a new iteration due to the Loop at Line 1.
By the fact that the state space is finite and that it starts a new iteration only if there is a progress in the ordering it is easy fo prove that $\hsolFun$ always terminates.
Observe that, when $\sttElm'\not\prec\sttElm$ where $\sttElm'$ is the state modified by $\MaxFun(\sttElm)$, it holds that $\sttElm'$ is also
strongly-maximal and, therefore, $\hsolFun$ returns a strongly maximal hybrid state.
At this point, we can go back to Line 3 of Algorithm~\ref{alg:hybsol} where the modified
state $\sttElm$ is strongly maximal as observed above. Now, by
Lemma~\ref{lmm:nxthalf} the input state of $\solFun$ at Line 5 is maximal.
Hence, by
external induction, the modified state composed of the results of Line 5 $\sttElm = ((\der{\regFun}, \prtElm), (\der{\undFun}, \cprtElm, \bElm[\PlrSym], \bElm[\OppSym]))$ is a hybrid state and, by Point 1) and 2) it is also greater or equal to the previous $\sttElm$, \wrt the ordering of
Definition~\ref{def:hybsrcspc}.
Similarly to the proof for $\hsolFun$, the state $\sttElm$ can be \emph{open} or \emph{closed}.
In the first case, two more cases may arise, indeed, if
$\sttElm$ is \emph{open}, then by Lemma~\ref{lmm:max} and Line 7, the function $\MaxFun$ modifies $\sttElm$ into a maximal hybrid state that is greater or equal to the previous $\sttElm$, \wrt the ordering of
Definition~\ref{def:hybsrcspc}.
On
the other hand, if $\sttElm$ is \emph{closed}, then by external induction it holds that
$\sttElm$ is promotable and, by Lemma~\ref{lmm:prm} the function $\PrmFun$ modifies $\sttElm$ into a
maximal hybrid state greater than the previous one the ordering of
Definition~\ref{def:hybsrcspc}.
Thus, at Line 11, is case $\sttElm \prec \der{\sttElm}$ where $\der{\sttElm}$ is the state stored before the recursive calls at Line 4, $\hsolFun$ is called on a maximal state and, as proved above, the modified resulting state is strongly maximal. Otherwise, as observer in the proof for $\hsolFun$, $\sttElm$ is already strongly-maximal.
As a consequence, by Lemma~\ref{lmm:und} and Line 10, we can conclude that
$\solFun$ returns a pair of region function and priority function
$(\der{\rFun},\der{\uFun})$ such that $\der{\sttElm} = ((\der{\regFun}, \cprtElm[\sttElm]), (\der{\undFun}, \qprtElm, \bSym[0],
      \bSym[1]))$ is a promotable hybrid state if $\der{\sttElm}$ is \emph{closed} and, an hybrid state otherwise, for all priority $\qprtElm \in \prtFun$ such that $\cprtElm[\sttElm] < \qprtElm \leq \min(\rng{{\regFun[\sttElm] \cup \undFun[\sttElm]}^{>\cprtElm[\sttElm]}})$, and $\bSym[0], \bSym[1] \in \SetN$.
\end{proof}

\section{Proof of Theorem \ref{theo:main}}
\main*
\label{app:mainth}

We provide an inductive proof over the highest priority.

For the \textbf{induction basis}, we consider, in this order, (1) $\solFun$ with highest priority $\bot$ (i.e.\ with an empty game); (2)
$\hsolFun$ being called with highest priority $0$; and
(3)
$\solFun$ being called with highest priority $0$.

For empty games---(1), highest priority $\bot$---there is nothing to do (and $\solFun$ does nothing).

For (2), games with maximal (and thus only) priority $0$ in $\hsolFun$, all states are in $R_s$, $s$ is closed, and all states in $R_s$ are promoted.
The game is then empty, $\solFun$ is called with an empty game, and Maximise is called on an empty set.
All global guarantees are retained due to Lemma \ref{lmm:prm}, case (1), and Lemma \ref{lmm:max}.
As $\mathsf{dom}(r_s^{\leq 0}) = \mathsf{dom}(u_s^{= 0}) = \emptyset$,
the additional requirement for $\hsolFun$ holds.

For (3), games with maximal (and thus only) priority $0$ in $\solFun$, all nodes are promoted---as we have shown for case (2)---in the first call of $\hsolFun$ (line 3).
$s$ is then empty. The call of $\solFun$ in line 5 therefore does nothing, and neither do the call of $\MaxFun$ in line 7, or of $\UndFun$ in line 10 for an empty game.

Thus, the global guarantees are retained by (2), (1), and due to Lemmas \ref{lmm:max}  and \ref{lmm:und}.

We now implement the \textbf{induction step}, again first for $\hsolFun$, and then for $\solFun$, using the results from $\hsolFun$.

For the induction step of $\hsolFun$,
we observe that, in each iteration of the loop, line 2 closes $R_s$ under attractor,
thus creating the precondition for the call of $\solFun$ (establishing the closure condition when $R_s$ is removed from the induced subgame on which the subcalls work).

When executed, the call of $\solFun$ then provides, by induction hypothesis, an increase of $\mathsf B_s$, which now contains every small dominion $D$ of player $1-\alpha$ of size $\leq \lfloor \frac{b_{1-\alpha_s}}{2} \rfloor$ in $S = \mathsf{dom}(r^{\leq p_s}) \setminus R_s$ from  before the call, while retaining all global guarantees.
Note that no such set exists if $\lfloor \frac{b_{1-\alpha_s}}{2} \rfloor = 0$ holds; while this case is not covered in the induction hypothesis, \textsf{sol} does nothing (or: returns immediately without implementing any change) in this case; thus, all global guarantees are retained.

Thus, the global guarantees are retained in every step either by induction hypothesis (Lemmas \ref{lmm:nxthalf}), by the special case $\lfloor \frac{b_{1-\alpha_s}}{2} \rfloor = 0$, or by Lemmas \ref{lmm:max} and Lemma \ref{lmm:prm}.

Finally, to see that $\mathsf B_s$ is closed, we observe that it is closed after $\MaxFun$ is run,
because closedness of $s$ always leads to a promotion, and thus to $s < \widehat s$, $s$ is open in the last iteration. Thus, the last operation executed in $\hsolFun$ is the call of $\MaxFun$, which entails closedness of $\mathsf B_s$ (Lemma \ref{lmm:max}).

We have shown that $\mathsf B_s$ now contains all small dominions $D$ of player $1-\alpha$ of size $\leq \lfloor \frac{b_{1-\alpha_s}}{2} \rfloor$ in $S = \mathsf{dom}(r^{\leq p_s}) \setminus R_s$ before the last call.
As $s$ has not changed during the call, it also contains all small dominions in $S = \mathsf{dom}(r^{\leq p_s}) \setminus R_s$ after the call, and on the time of return.
As $R_s$ is a
$p_s$-region,
Lemma \ref{lem:no-p} provides that, for any non-empty dominion $D'$ in $\mathsf{dom}(r^{\leq p_s})$,
there must be a non-empty sub-dominion $D'' \subseteq D'$, which is also a dominion in $S$.
Thus, as $S$ does not contain a dominion $D$ of player $1-\alpha_s$ with size $\leq \lfloor \frac{b_{1-\alpha_s}}{2} \rfloor$, neither does $\mathsf{dom}(r^{\leq p_s})$ when $\hsolFun$ returns.

We can now turn to the induction step for $\solFun$.
We fist check that the global guarantees are retained up to the point where $\UndFun$ is called.
We then make the argument that $\solFun$ will, for a dominion $D$ of player $1-\alpha$ in $\mathsf P_s$ of size $|D| \leq b_{1-\alpha}$, guarantee that $D \subseteq \mathsf B_s$ holds.
We will then show that, under these conditions, $\UndFun$ will retain the global guarantees.

The check that the global guarantees are retained up to the point where $\UndFun$ is called is straight forward: we have seen that $\hsolFun$ (line 3) does, and also provides that $R_s$ is closed under attractor before $\solFun$ is called in line 5.

$\MaxFun$ and $\PrmFun$ do retain the global guarantees by Lemmas \ref{lmm:max} and \ref{lmm:prm}, respectively.

We now argue that, when $\UndFun$ is called,
$\mathsf B_s$ contains all dominions of player $1-\alpha_s$ in $\mathsf P_s$ of size $\leq b_{1-\alpha_s}$.
Let $D$ be such a dominion.

We first observe that $\mathsf G_s$, and its $\alpha$-attractor $\atrFun[][\alpha][\mathsf G_s]$, cannot intersect with $D$ at any point, by the global guarantees.

Thus, $D$ is a dominion in $\mathsf P_s \setminus \atrFun[][\alpha][\mathsf G_s]$ at any point by Lemma \ref{lem:noattraction}.

After the first call of $\hsolFun$ (in line 3), we have established that $\mathsf{dom}(r^{\leq p_s})$ does not contain a dominion of player $1-\alpha_s$ of size $\leq \lfloor \frac{b_{1-\alpha_s}}{2}\rfloor$.

In particular, if $\mathsf{dom}(r^{\leq p_s})$ does not intersect with $D$, then $D$ must be contained in $\mathsf B_s$, as $\mathsf P_s \subseteq \mathsf{dom}(r^{\leq p_s}) \cup \mathsf B_s \cup \mathsf G_s$.

Otherwise, $D' = D \cap \mathsf{dom}(r^{\leq p_s})$ must be a dominion in $\mathsf{dom}(r^{\leq p_s})$ (with the same winning strategy for player $\alpha$ as in $\mathsf P_s$ (or $\mathsf P_s \setminus \atrFun[][\alpha][\mathsf G_s]$) by Lemma \ref{lem:noattraction}.

By Lemma \ref{lem:no-p}, $D'$ has a sub-dominion $D''\subseteq D'$ which does not intersect with $R_s$.
$|D''| > \lfloor \frac{b_{1-\alpha_s}}{2}\rfloor$ follows from the fact that $\mathsf{dom}(r^{\leq p_s})$ contains no smaller dominions.
Thus, by induction hypothesis, after the call of $\solFun$ from line 5 is returned, $\mathsf B_s$ also contains $D''$.

In this case, $s< \widehat s$ in line 9, such that $\hsolFun$ is called.
In this case, also calling $\MaxFun$ (by Lemma \ref{lmm:max} or $\PrmFun$ (By Lemma \ref{lmm:prm}) retain the global guarantees.

Let us assume for contradiction that, after the return from $\hsolFun$,
$I = \mathsf{dom}(r^{\leq p_s}) \cap (D' \setminus D'')$ is not empty.
Noting that $D \supseteq D' \supseteq D' \setminus D''$, $|D| \leq b_{1-\alpha_s}$, and $|D''|>\lfloor \frac{b_{1-\alpha_s}}{2}\rfloor$, $|I| \leq \lfloor \frac{b_{1-\alpha_s}}{2}\rfloor$.

Moreover, we have again that $I$ is contained in $\mathsf P_s \setminus \atrFun[][\alpha][\mathsf G_s]$ by our global guarantees, $\mathsf P_s \subseteq \mathsf{dom}(r^{\leq p_s}) \cup \mathsf B_s \cup \mathsf G_s$, and $\mathsf B_s$ is closed under $1-\alpha_s$ attractor in $\mathsf P_s$.

Thus, $I = (D' \setminus D'')\cap \mathsf{dom}(r^{\leq p_s})$
must be a dominion in $\mathsf{dom}(r^{\leq p_s})$ (with the same winning strategy for player alpha as in $\mathsf P_s$ (or $\mathsf P_s \setminus \atrFun[][\alpha][\mathsf G_s]$) by Lemma \ref{lem:no-p}. (Contradiction between the assumption that $I$ is non-empty, and that $\mathsf{dom}(r^{\leq p_s})$ cannot contain such a small dominion of player $1-\alpha_s$).

We have now established that the global guarantees hold, while $\mathsf B_s$ contains all dominions of size $\leq b_{1-\alpha_s}$ of player $1-\alpha_s$ in $\mathsf P_s$, when $\UndFun$ is called.
This entails that $\mathsf G_s \cup L_s$ contains all dominions of size $\leq b_{\alpha_s}$ of player $\alpha$ in $\mathsf P_s$, when $\UndFun$ is called.

$\UndFun$ then moves the respective set, $U_s$ or $L_s$, to $u_s^{-1}(c_s)$ of the caller priority, such that $\mathsf B_{s'}$ on that level is the current $\mathsf B_s$ or $\mathsf G_s \cup L_s$ when $c_s \equiv_2 p_s$ or $c_s \not\equiv_2 p_s$, respectively.

Thus, due to Lemma \ref{lmm:und}, $\UndFun$ retains the global guarantees.

This completes the induction step.
% End of file AppendixA.tex

%% file: AppendixC.tex
% Begin of file AppendixC.tex

%appendix for the AUXILIARY functions/procedures of RPP, Parys, and HPP
\section{Appendix for Section \ref{sec:hybalg;sub:prtprm}}
\label{app:c}

\begin{figure}[t]
\vspace{-1em}
\noindent
\begin{minipage}[t]{0.505\textwidth}
  \null
  \algsolexp
\end{minipage}
\begin{minipage}[t]{0.495\textwidth}
  \null
  \algsolexpmacro
\end{minipage}
\vspace{-2.5em}
\end{figure}

The auxiliary function $\NxtFun$ generates a new maximal state $\der{\sttElm} =
\NxtFun(\sttElm) \in \SttSet[\MSym]$, starting from a strongly-maximal one
$\sttElm \in \SttSet[\SSym]$.
The state $\der{\sttElm}$ is obtained by changing the current priority
$\prtElm[\sttElm]$ to the highest priority $\qprtElm$ of the positions in
$\LFun[\sttElm] \setminus \RFun[\sttElm]$.
Observe that when no such position exists, namely when $\LFun[\sttElm] =
\RFun[\sttElm]$, the new priority coincides with $\bot$.

$\MaxFun$ enforces the maximality property on the state $\sttElm$ received as
input, so that, in the resulting state obtained by modifying $\sttElm$ in-place,
no position of the local area $\LSet[\sttElm]$ can be attracted by the quasi
$\alpha$-dominions $\HSet[\sttElm][\alpha]$, with $\alpha \in \SetB$.
To this end, the procedure computes at Line~2 the $\alpha$-attractor
$\atrFun[][\alpha][ {\HSet[\sttElm][\alpha] \setminus \LSet[\sttElm],
\LSet[\sttElm]} ]$, collecting all the position of $\LSet[\sttElm]$ that player
$\alpha$ can force to move into $\HSet[\sttElm][\alpha] \setminus
\LSet[\sttElm]$.
The minimum priority $\qprtElm$ assigned by $\regFun$ to a position in the
attracting set is extracted at Line~3.
The attracted positions are then assigned priority $\qprtElm$ in the region
function $\regFun[\sttElm]$ at Line~4.
Since removing positions from the local area $\LSet[\sttElm]$ may induce a
violation of the two requirements of Definition~\ref{def:regfun}, the positions
that remain in $\LSet[\sttElm]$ at the end of the for-each loop of Lines~1-4
need to be reset to their original priority, as prescribed at Line~5.

To conclude, the procedure $\PrmFun$ requires a promotable state $\sttElm \in
\SttSet[\PSym]$ and applies a promotion operation to the region
$\RFun[\sttElm]$, while preserving any maximality property already enjoyed by
the input state.
It first computes the opponent best-escape priority $\qprtElm$ for the set
$\RFun[\sttElm]$ \wrt $\regFun[\sttElm]$ (Line~1).
Intuitively, this is the smallest priority the opponent can reach with one move
when escaping from the region $\RFun[\sttElm]$.
Formally, it is defined as:
\[
  \bepFun[][ {\dual{\alpha}[\sttElm]} ][{\RFun[\sttElm]}, {\regFun[\sttElm]}]
  \defeq \min(\rng{\regFun[\sttElm] \rst \rng{\IRel}}),
\]
where $\IRel \defeq \MovRel[\GameName] \cap (\escFun[][ {\dual{\alpha}[\sttElm]}
](\RFun[\sttElm]) \times (\dom{\regFun[\sttElm]} \setminus \RFun[\sttElm]))$
contains all the moves leading outside $\RFun[\sttElm]$ that the opponent can
use to escape.
The procedure, then, promotes $\RFun[\sttElm]$ to $\qprtElm$, by assigning at
Line~2 the priority $\qprtElm$ to all the positions of $\RFun[\sttElm]$ in the
region function $\regFun[\sttElm]$.
Observe that, thanks to the $\dual{\alpha}[\sttElm]$-maximality of the input
state, $\qprtElm$ is necessarily congruent to $\alpha_{\sttElm}$.
In particular, when the only possibility for player $\dual{\alpha}[\sttElm]$ to
escape from $\RSet[\sttElm]$ is to reach
$\regFun[\sttElm][-1](\topSym[\alpha])$, the value of $\qprtElm$ is
$\topSym[\alpha]$.
In this case, we are promoting $\RSet[\sttElm]$ from the status of quasi
$\alpha_{\sttElm}$-dominion to that of $\alpha_{\sttElm}$-dominion.
The correctness of this is ensured by Theorem~\ref{thm:inddom}.

At this point, by defining $\solFun(\GameName) \defeq (\regFun[][-1](\topZSym),
\regFun[][-1](\topOSym))$, where $\regFun \defeq \solFun(\isttElm)$, we obtain a
sound and complete solution algorithm for parity games.
In particular, the soundness follows from the fact that \RPP always traverses
states having as invariant the property that $\regFun[][-1](\topZSym)$ and
$\regFun[][-1](\topOSym)$ are dominions (see Item~\ref{def:qsidomfun(dom)} of
Definition~\ref{def:qsidomfun}).
Completeness, instead, is due to the recursive nature of the algorithm, whose
base case ensures that no position is left unprocessed at any given priority.

\begin{figure}[t]
\vspace{-1em}
\noindent
\begin{minipage}[t]{0.52\textwidth}
  \null
  \algsolpar
  \algsolparmacroi
\end{minipage}
\hspace{0.15em}
\noindent
\begin{minipage}[t]{0.45\textwidth}
  \null
  \algsolparhalf
  \vspace{0.35em}
  \algsolparmacroii
\end{minipage}
\vspace{-0.5em}
\end{figure}

The auxiliary function $\NxtFun$ of the Parys Solver generates a new state by decreasing the current priority. This function is also called by another auxiliary function, the $\hsolFun$, that generates a new state with halved precision for the player $\alpha$ and decreased priority.

As for the \RPP solver, $\MaxFun$ enforces the maximality property on current state $\sttElm$. As a result, in the modified state no position of the local area $\LSet[\sttElm]$ can be attracted by the quasi
$\alpha$-dominions $\HSet[\sttElm][\alpha]$.
To this end, the procedure computes at Line~1 the $\dual{\alpha}$-attractor
$\atrFun[][\dual{\alpha}][ {\HSet[\sttElm][\dual{\alpha}],
\LSet[\sttElm]} ]$, collecting all the position of $\LSet[\sttElm]$ that player
$\dual{\alpha}$ can force to move into $\HSet[\sttElm][\dual{\alpha}]$.
The attracted set is then assigned to $\regFun$ to a position in the
attracting set is extracted at Line~3.
The attracted positions are then assigned to $\undFun$ with the current priority $\prtElm[\sttElm]$.

Finally, the $\UndFun$ function, receive in input a state where all small dominions of player $\dual{\alpha}[\sttElm]$ of size $\leq b_{\dual{\alpha}[\sttElm]}$ \emph{from the time of the call} are
processed.
Therefore, none of this small dominions intersect with $\LSet[\sttElm]$, that is in turn moved to the previous priority level. The positions of
$\USet[\sttElm]$ can be instead reset in $\regFun$ to their original priority.

The auxiliary function $\NxtFun$ generates a new maximal state $\der{\sttElm} =
\NxtFun(\sttElm) \in \SttSet[\MSym]$, starting from the strongly-maximal one
$\sttElm \in \SttSet[\SSym]$ in input.
The state $\der{\sttElm}$ is obtained by first setting the caller priority to
the current one $\prtElm[\sttElm]$ and, then, by computing the highest priority
$\qprtElm$ among the unprocessed positions in $\LFun[\sttElm] \setminus
\RFun[\sttElm]$.
The promotion and maximisation procedures, $\PrmFun$ and $\MaxFun$, generalise
the corresponding ones associated with \RPP.
The only difference is that here we need to determine which one, between the
region function $\regFun$ and the promotion function $\undFun$, has to receive
the promoted region $\RegSet[\sttElm]$, in case of $\PrmFun$, or the positions
attracted from $\LSet[\sttElm] \cup \USet[\sttElm]$, in case of $\MaxFun$.
As before, $\PrmFun$ asks for the input state to be promotable, \ie, $\sttElm
\in \SttSet[\PSym]$, while $\MaxFun$ does not require any specific property on
it.

\vspace{-1em}
\noindent
\begin{minipage}[t]{0.496\textwidth}
  \null
  \algsolqsipol
\end{minipage}
\hspace{-0.5em}
\begin{minipage}[t]{0.504\textwidth}
  \null
  \algsolqsipolhalf
\end{minipage}

\noindent
\begin{minipage}[t]{0.475\textwidth}
  \null
  \algsolqsipolmacroi
\end{minipage}
\hspace{-0.5em}
\begin{minipage}[t]{0.535\textwidth}
  \null
  \algsolqsipolmacroii
\end{minipage}

Similarly to Parys' algorithm, the $\HalfFun$ function halves the bound of the
opponent player $\dual{\alpha}[\sttElm]$, leaving the bound of player
$\alpha_{\sttElm}$ unchanged.
Finally, the $\UndFun$ function, starts from a strongly maximal state $\sttElm
\in \SttSet[\SSym]$ where all small dominions of player $\dual{\alpha}[\sttElm]$
of size $\leq b_{\dual{\alpha}[\sttElm]}$ \emph{from the time of the call} are
processed.
Thus, neither do any of the small dominions ($\leq b_{\dual{\alpha}[\sttElm]}$)
of player $\dual{\alpha}[\sttElm]$ intersect with $\USet[\sttElm]$, nor do any
of the small dominions ($\leq b_{\dual{\alpha}[\sttElm]}$) of player
$\dual{\alpha}[\sttElm]$ intersect with $\LSet[\sttElm]$.
Depending on the parity of the calling priority, we can then return the
respective set and, where the parity is different, reset the positions of
$\USet[\sttElm]$ in $\regFun$ to their original priority.

% End of file AppendixC.tex

%% file: ExperEval.tex
% Begin of file ExperEval.tex

\section{Experimental Evaluation - Keiren's Benchmarks}
\label{sec:expereval}

In this set of benchmarks we consider were first proposed in~\cite{Kei15} and
comprises a number of concrete verification problems, ranging from
model-checking, to equivalence-checking and decision problems for different
temporal logics.
They can be divided in the following four categories.

\noindent\textbf{\textit{Model-checking benchmarks.}}
The first group contains 313 games, with size up to $O(10^{7})$ positions.
It includes a number of different verification problems.
A first set contains encodings of a variety of communication protocols
from~\cite{KM90,CK74,GP96,BSW69}: the alternating bit protocol, the positive
acknowledgement with retransmission protocol, the bounded retransmission
protocol, and the sliding window protocols.
The protocols are parameterised with the number of messages to send and, when
applicable, the window size.
The set also contains verification problems for the cache coherence protocol
of~\cite{VHBJB01} and the wait-free handshake register of~\cite{Hes98}, as well
as the classic elevator and towers-of-Hanoi benchmarks from~\cite{FL09}.
The verification tasks under analysis cover fairness, liveness and safety
properties.
A second set, instead, contains encodings of two-player board games, such as
Clobber, Domineering, Hex, Othello, and Snake, all parameterised by their board
size.
Here, the existence of a winning strategy for the game is the property
considered.
The encoding into parity games results in games with very few priorities: up to
4 in some cases.
%3 priorities in most of the cases, with 5 priorities in rare ones.

\noindent\textbf{\textit{Equivalence checking benchmarks.}}  This group contains 216 games
encoding equivalence tests between processes.  The verification problems test
various forms of process equivalences, such as strong, weak and branching
bisimulation, as well as branching simulation.  Most of the processes are the
ones already considered in the model-checking benchmarks.  The encoding into
parity games results in games with at most two priorities, hence the only
relevant measure of difficulty is the size, again reaching $O(10^{7})$ nodes for
the bigger instances.

%\vspace{1.em}
\begin{table}
  \vspace{-1.em}
  \begin{center}
    \small
    \scalebox{0.89}[0.90]{\tabper}
  \end{center}
  \caption{\label{tab:per} \small Solution times in seconds on Keiren's
  benchmarks (1012 games).}
  \vspace{-2em}
\end{table}
%\vspace{-1.em}

\noindent\textbf{\textit{Decision problem benchmarks.}}
The third group contains encodings of satisfiability and validity problems for
formulae of various temporal logics: \LTL, \CTL, \CTLS, \PDL and the \MC, and
comprises 192 games.
The maximal size of a benchmark is around $3\cdot 10^6$ positions.
The parity games encoding have been obtained with the tool
MLSolver~\cite{FL10a}.
The situation here is more interesting, since these concrete problems feature a
higher number of priority, up to 20 in few cases.
Hence, unlike the previous two groups, these benchmarks allow us to stress a bit
more the scalability of the solution algorithms \wrt the increase in priorities.

\noindent\textbf{\textit{PGSolver.}}
This group contains 291 synthetic benchmarks, corresponding to known families of
hard cases for specific solvers and randomly generated ones.  The sizes and
number of priorities vary significantly, depending on the specific class of
games.

Table~\ref{tab:per} reports the results of the experiments for all the solvers
considered in the analysis, divided by class of benchmarks\footnote{The
benchmarks were run by issuing the following OINK commands: oink --no-single --no-loops --no-wcwc $GameName$ $SolverName$; where solvers are: ZLK, NPP, SSPM, QPT, ZLKQ.}.
For each solver, the total completion time, the average time per benchmark and
the percentage of timed-out executions are given.
We set a timeout of 10 seconds for all the benchmarks, except for the
equivalence-check class, for which 40 seconds is used instead.
As expected, the exponential solvers perform better on all the classes, with
\PP taking the lead most of the time.
\SSPM and \QPT both perform quite poorly, between two and three orders of
magnitude worse than the other solvers, and do not seem to scale beyond the
simplest instances, as also evidenced by the high number of timeouts.
Both $Par$ and \HPP, instead, perform relatively well in all the benchmarks,
being able to solve all the instances without incurring in timeouts and
maintaining a short distance from the exponential solvers performance-wise. $Par$ has a slight edge over \HPP on the model-checking and equivalence checking
problems, both of which feature a very low number of priorities, though the time
advantage on average is typically negligible.
On the other hand, when the number of priorities increases, like in the decision
problems, the situation reverses and \HPP takes the lead over $Par$ and
practically matches the performance of the exponential solvers.
This seems to suggest that \HPP may scale better \wrt the number of priorities
in the games.
To further investigate this behaviour we decided to perform additional
experiments, whose results are reported in the next subsection.